\documentclass{amsart}[11pt]
\usepackage{amsmath, amsfonts, amsthm, amssymb}
\usepackage{subfigure}
\usepackage{verbatim}
\usepackage{hyperref}
\usepackage{color}
\usepackage{enumerate}
\usepackage[dvips]{graphicx}
\numberwithin{equation}{section}
\newtheorem{theorem}{Theorem}[section]
\newtheorem{corollary}[theorem]{Corollary}
\newtheorem{lemma}[theorem]{Lemma}
\newtheorem{proposition}[theorem]{Proposition}

\theoremstyle{definition}
\newtheorem{definition}[theorem]{Definition}
\newtheorem{remark}[theorem]{Remark}

\newtheorem{assumption}[theorem]{Assumption}

\newcommand{\ind}{1\hspace{-2.1mm}{1}}
\newcommand{\Ss}{\mathcal{S}}

\newcommand{\D}{\mathrm{d}}

\newcommand{\RR}{\mathbb{R}}
\newcommand{\Ll}{\mathcal{L}}
\newcommand{\Pp}{\mathcal{P}}
\newcommand{\Cc}{\mathcal{C}}
\newcommand{\Mm}{\mathcal{M}}
\newcommand{\Nn}{\mathcal{N}}
\newcommand{\BS}{\mathrm{BS}}
\newcommand{\E}{\mathrm{e}}
\newcommand{\eps}{\varepsilon}
\newcommand{\Ff}{\mathcal{F}}
\newcommand{\Zz}{\mathcal{Z}}
\newcommand{\EE}{\mathbb{E}}
\newcommand{\loc}{\mathrm{loc}}

\begin{document}

\title{Generalised arbitrage-free SVI volatility surfaces}
\author{Gaoyue Guo}
\address{Ecole Polytechnique Paris}
\email{guo.gaoyue@gmail.com}
\author{Antoine Jacquier}
\address{Department of Mathematics, Imperial College London}
\email{a.jacquier@imperial.ac.uk}
\author{Claude Martini}
\address{Zeliade Systems, Paris}
\email{cmartini@zeliade.com }
\author{Leo Neufcourt}
\address{Department of Statistics, Columbia University}
\email{ln2294@columbia.edu}
\thanks{The authors are indebted to Stefano De Marco for numerous remarks and comments in the early stages of this paper, and would like to thank the anonymous referees for their precise and helpful suggestions.
AJ acknowledges financial support from the EPSRC First Grant EP/M008436/1.}
\date{\today}
\begin{abstract}
In this article we propose a generalisation of the recent work by Gatheral-Jacquier~\cite{gj} 
on explicit arbitrage-free parameterisations of implied volatility surfaces. 
We also discuss extensively the notion of arbitrage freeness and Roger Lee's moment formula using the recent analysis by Roper~\cite{roper}.
We further exhibit an arbitrage-free volatility surface different from Gatheral's SVI parameterisation.
\end{abstract}

\maketitle


\section{Introduction}
European option prices are usually quoted in terms of the corresponding implied volatility, 
and over the last decade a large number of papers (both from practitioners and academics) 
has focused on understanding its behaviour and characteristics.
The most important directions have been towards (i) understanding the behaviour of the implied volatility in a given model~\cite{BF, BBF, FGGS, GuSt} and (ii) deciphering its behaviour in a model-independent way, 
as in~\cite{Lee, roper, RT10}.
These results have provided us with a set of tools and methods to check
whether a given parameterisation is free of arbitrage or not. 
In particular, given a set of observed data (say European Calls and Puts for different strikes and maturities), 
it is of fundamental importance to determine a methodology ensuring that both 
interpolation and extrapolation of this data are also arbitrage-free.
Such approaches have been carried out for instance in~\cite{CarrWu, Fengler, Zeliade}.
Several parameterisations of the implied volatility surface  have now become popular, 
in particular~\cite{Gatheral:2004, Hagan, HypHyp}, albeit not ensuring absence of arbitrage.

Recently, Gatheral and Jacquier~\cite{gj} proposed a new class of implied volatility parameterisation, 
based on the previous works by Gatheral~\cite{Gatheral:2004}.
In particular they provide explicit sufficient and---in a certain sense---almost necessary conditions ensuring that such a surface is free of arbitrage.
We shall recall later the exact definition of arbitrage, and see that it can be decomposed into 
two elements: butterfly arbitrage and calendar spread arbitrage.
This new class depends on the maturity and can hence be used to model the whole volatility surface,
and not a single slice.
It also depends on the at-the-money total implied variance $\theta_t$,
and on a positive function $\varphi$ such that the total variance~$w$ as a function of time-to-maturity
$t$ and log-(forward)-moneyness $k$ is given by
$w(k,t)\equiv \theta_t \text{SVI}_{\rho}(k \varphi(\theta_t) )$,
where $\text{SVI}_{\rho}$ is the classical (normalised) SVI parameterisation from~\cite{gj},
and~$\rho$ an asymmetry parameter (essentially playing the role of the correlation between spot and volatility in stochastic volatility models).

In this work, we generalise their framework to volatility surfaces parameterised as
$w(k,t)\equiv \theta_t \Psi(k \varphi(\theta_t))$ for some (general) functions~$\varphi$, $\theta$, $\Psi$.
We obtain (Sections~\ref{sec_cal} and~\ref{sec_butter}) necessary and sufficient conditions coupling 
the functions~$\Psi$ and~$\varphi$ that preclude arbitrage.
This allows us to obtain
(i) the exact set of admissible functions $\varphi$ in the symmetric ($\rho=0$) SVI case, 
and (ii) a constraint-free parameterisation of Gatheral-Jacquier functions satisfying the conditions of~\cite{gj}.
In passing (Section~\ref{sec:nonsmooth}), we extend the class of possible functions by allowing for non-smooth implied volatility functions.
Finally (Section~\ref{sec_example}), we exhibit examples of non-SVI arbitrage-free implied volatility surfaces.


\textbf{Notations:}
We consider here European option prices with maturity $t\geq 0$ and strike $K\geq 0$, 
written on an underlying stock $S$. 
Without loss of generality we shall always assume that $S_0=1$ and that interest rates are null, 
and hence the log (forward) moneyness reads $k:=\log(K)$.
We denote by 
\begin{equation}\label{eq:BlackScholes}
\BS(K,w) = \Nn(d_+(\log(K), w)) - K \Nn(d_-(\log(K),w)),
\end{equation}
the Black-Scholes value for a European Call option 
with strike~$K$ and total variance~$w$,
where $\Nn$ denotes the Gaussian cumulative distribution function and 
$d_\pm(k,w) := -k/\sqrt{w}\pm \sqrt{w}/2$;
more generally, we shall write~$\mathrm{C}(K, t)$ for (any) European Call prices with strike~$K$ and maturity~$t$.
For any $k\in\RR, t\geq 0$, the corresponding implied volatility is denoted by $\sigma(k,t)$ and the total variance~$w$ is defined by
$w(k,t):=\sigma(k,t)^2 t$.
With a slight abuse of language (commonly accepted in the finance jargon), 
we refer to the two-dimensional map $(k,t)\mapsto w(k,t)$ as the (implied) volatility surface.
Finally, for two functions~$g$ and~$h$ not null almost everywhere, we say that $g(z)\sim h(z)$ at $z=0$ whenever $\lim_{z\to 0}g(z)/h(z)=1$.
We shall also use the notations $\RR_+:=[0,\infty)$ and $\RR^*_+:=(0,\infty)$,
and use the convention $\inf_{\emptyset}=\infty$.

\section{Absence of arbitrage and volatility parameterisations}\label{sec:Arbitrage}
This preliminary section serves several purposes: 
we first recall the very definition of `arbitrage freeness' and its characterisation in terms of implied volatility.
We then state and prove a few results (which are also of independent interest) related to this notion of arbitrage.
We finally quickly review the parameterisation proposed in~\cite{gj} and introduce an extension, which is our new contribution.

\subsection{Absence of arbitrage}
As defined in~\cite{Cox}, absence of static arbitrage corresponds to the existence of a non-negative local martingale 
(on some probability space) such that European Call options (on this local martingale) 
can be written as risk-neutral expectations of their final payoffs.
Armed with this definition, it is however not easy to check whether a given set of (Call) option prices
yields an arbitrage or not.
A more practical route follows Roper's~\cite{roper} arguments (or equivalently~\cite{gj}), 
who provide sufficient and almost\footnote{The `almost' refers to~\cite[Theorem 2.15]{roper},
where smoothness and strict positivity of the implied volatility are required.} necessary conditions for a given two-dimensional function 
(of strike and maturity) to be a proper implied volatility surface, i.e. to generate arbitrage-free European option prices.
Note that Cox and Hobson's definition~\cite{Cox} allows for strict local martingales, 
whereas Roper's framework only considers true martingales, his argument being that the implied volatility
is ill-defined for strict local martingales, in particular through the failure of Put-Call parity.
Following collateralisation arguments developed in~\cite{Cox}, the recent paper~\cite{JacquierKR}
restores Put-Call parity in strict local martingale models 
and clarifies the definition and properties of the implied volatilities 
(differently generated from Put and from Call options).
Pursuing the goal set up in~\cite{gj}, we shall exclude here in our modelling framework the strict local martingale case,
and understand `static arbitrage' as a restriction to true martingales.\footnote{For 
a true martingale~$\Ss$, 
it is easy to see that, for a fixed maturity~$T$, 
the map $K\mapsto \EE(\Ss_T - K)_+$ is decreasing, convex and tends to zero at infinity, 
properties that still hold in the strict local martingale setting.
However, as shown by Pal and Protter~\cite{Pal}, Call prices are not necessarily increasing in maturity
in strict local martingale models, and therefore the corresponding total implied variance, 
whenever defined, need not be an increasing map any longer.}
We now define these terms precisely, and refer to~\cite{roper} for full details.
\begin{definition}\label{def:NoArb}
Given a map $(K,t)\in \RR_+ \times \RR_+\mapsto C(K,t)$, we say that there is no static arbitrage if there exists a non-negative martingale $S$ on some filtered probability space $(\Omega, \Ff ,(\Ff_t)_{t\geq 0}, \mathbb{P})$ such that
$C(K, t) = \EE((S_t-K)_+\vert \Ff_0)$ for each $(K,t) \in \RR_+\times \RR_+$.
\end{definition}
Consider now a two-dimensional map $w:\RR\times \RR_+\to \RR_+$ representing a total variance surface;
it is then natural to wonder whether the Call price surface defined by 
$\RR_+ \times \RR_+ \ni (K,t) \mapsto \BS(K, w(\log(K),t))$ is free of static arbitrage.
Introduce the operator $\Ll$ acting on $\Cc^{2,1}(\RR\times \RR_+^*\to \RR_+^*)$ functions by
\begin{equation}\label{eq:OpL}
\Ll w (k,t):=\left(1-\frac{k\partial_kw(k,t)}{2w(k,t)}\right)^2-\frac{(\partial_kw(k,t))^2}{4}\left(\frac{1}{w(k,t)}+\frac{1}{4}\right)+\frac{\partial_{kk}^2w(k,t)}{2},
\quad\text{for all }k\in\RR, t>0.
\end{equation}
Note that even though $\Ll$ does not act on the second component of the function, 
we shall keep this notation for clarity.
For fixed $t>0$, the total variance~$w(k,t)$ may in principle be null for some $k\in \RR$,
which might break the well-posedness of the right-hand side of~\eqref{eq:OpL}.
However, it is easy to show that~$w(k, t)$ is strictly positive whenever~$k$ 
belongs to the support of the log stock price at time~$t$,
and the restriction $w(k,t)>0$ is therefore sensible, which is imposed in model~\eqref{eq:GenSVI} 
with Assumption~\ref{assu:SVI}(iii).
At $t=0$, the total variance is equal to zero everywhere, and the definition of the operator~$\Ll$
shall not be needed.
Roper~\cite[Theorem 2.9]{roper} proved the following theorem:
\begin{theorem}\label{thm:RoperNoArb}
If the two-dimensional map~$w:\RR\times \RR_+\to\RR_+$ satisfies
\begin{enumerate}[(i)]
\item $w(\cdot,t)$ is of class $\Cc^2(\RR)$ for each $t\geq 0$;
\item $w(k,t)>0$ for all $(k,t)\in \RR\times \RR^*_+$;
\item $w(k,\cdot)$ is non-decreasing for each $k\in\RR$;
\item for each $(k,t)\in \RR\times \RR^*_+$, $\Ll w (k,t)$ is non-negative;
\item $w(k, 0)=0$ for all $k\in\RR$;
\item $\lim_{k\uparrow \infty}d_+(k,w(k,t)) = -\infty$, for each $t>0$.
\end{enumerate}
Then the corresponding Call price surface $(K,t)\mapsto \BS(K,w(\log(K),t))$ is free of static arbitrage.
\end{theorem}
Conditions (i),  (ii) and (v) are usually easy to check.
The other conditions motivate the following weaker notions of arbitrage commonly used in practice, in the maturity and in the strike directions:
\begin{definition}\label{def:WeakArb}
Let $w:\RR\times\RR_+^*\to\RR_+$ be a two-dimensional map satisfying Theorem~\ref{thm:RoperNoArb}(i)-(ii).
\begin{itemize}
\item $w$ is said to be free of calendar spread arbitrage if Condition~(iii) in Theorem~\ref{thm:RoperNoArb} holds;
\item $w$ is said to be free of butterfly arbitrage if Condition~(iv) in Theorem~\ref{thm:RoperNoArb} holds.
\end{itemize}
\end{definition}
Butterfly arbitrage corresponds to the convexity of option prices, 
which can be read as a condition on the behaviour of the implied volatility surface
(\cite[Definition 2.3]{gj} and~\cite[Lemma 2.2]{gj}).
If $\sigma_{\loc}$ represents the (Dupire) local volatility, the relationship
$\sigma^2_{\loc}(k,t) = \partial_t w(k,t)/\Ll w (k,t)$, for all $k\in\mathbb{R}, t>0$ is now standard (see~\cite[Chapter 1, Equation (1.10)]{jimbook}).
Therefore absence of static arbitrage implies that both the numerator and the denominator are non-negative quantities.
Condition~(vi) in Theorem~\ref{thm:RoperNoArb} is called the `Large-Moneyness Behaviour' (LMB) condition,
and is equivalent to Call option prices tending to zero as the strike tends to (positive) infinity,
as proved in~\cite[Theorem 5.3]{RT10}.
The following lemma however shows that other asymptotic behaviours of~$d_+$ and~$d_-$ 
hold in full generality.
This was proved by Rogers and Tehranchi~\cite{RT10} in a general framework, 
and we include here a short self-contained proof.

\begin{lemma}\label{lem:RT10}
Let $w$ be any positive real function. Then
\begin{enumerate}[(i)]
\item $\lim_{k\uparrow \infty} d_-(k, w(k))=-\infty$;
\item $\lim_{k\downarrow -\infty} d_+(k, w(k))=+\infty$.
\end{enumerate}
\end{lemma}
\begin{proof}
The arithmetic-geometric mean inequality reads
$-d_-(k, w(k)) =  \frac{k}{\sqrt{w(k)}}+\frac{\sqrt{w(k)}}{2}\geq \sqrt{2k}$, when $k>0$, 
which implies (i), and (ii) follows using
$d_+(k, w(k)) =  \frac{-k}{\sqrt{w(k)}}+\frac{\sqrt{w(k)}}{2}\geq \sqrt{-2k}$, when $k<0$.
\end{proof}

The missing statements in Lemma \ref{lem:RT10} are the LMB Condition (Condition~(vi) in Theorem~\ref{thm:RoperNoArb}) and
the Small-Moneyness Behaviour (SMB):
$\lim_{k\downarrow -\infty} d_-(k, v(k))=+\infty$.
To investigate
further, let us remark that the framework developed in~\cite{roper} encompasses situations where the underlying stock price can be null with positive probability. 
This can indeed be useful to model the probability of default of the underlying. 
Computations similar in spirit to~\cite{roper} show that the marginal law of the stock price at some fixed time $t>0$ 
has no mass at zero if and only if
$\lim_{K \downarrow 0}\partial_K\mathrm{C}(K, t) = -1$,
which is a statement about a 'small-moneyness' behaviour.
This can be fully recast in terms of implied volatility, and the above missing conditions then come naturally into play
in the following proposition, the proof of which is postponed to Appendix~\ref{sec:prop:nomass}:

\begin{proposition} \label{prop:nomass} (Symmetry under small-moneyness behaviour)
Let~$v$ be a $\Cc^2(\RR)$ real function satisfying
\begin{enumerate}[(I)]
\item $v(k)>0$ and $\Ll v (k) \geq 0$ for all $k \in \mathbb{R}$;
\item $\lim_{k\downarrow -\infty} d_-(k, v(k))=+\infty$ (SMB Condition);
\item $\lim_{k \uparrow \infty} d_+(k, v(k)) = -\infty$ (LMB Condition).
\end{enumerate}
Define the two functions~$p_-$ and~$p_+$ by
$k\mapsto p_{\pm}(k) := \left(2 \pi v(k)\right)^{-1/2}\exp\left(-\frac{1}{2}d_{\pm}^2(k,v(k))\right) \Ll v(k)$.
Then
\begin{enumerate}
\item $p_{+}$ and $p_-$ define two densities of probability measures on $\RR$ with respect to the Lebesgue measure, 
i.e. $\int_{\RR} p_-(k) \D k=\int_{\RR} p_+(k) \D k=1$;
\item $p_+(k) = \E^{k} p_-(k)$, so that $\int_{-\infty}^\infty \E^{k} p_-(k) \D k=\int_{-\infty}^\infty \E^{-k} p_+(k) \D k= 1$;
\item $p_-$ is the density of probability associated to Call option prices with implied volatility $v$,
in the sense that 
$p_-(k)\equiv \E^{k}\partial^2_{KK} \BS(K, v(\log(K)))|_{K=\E^{k}}$, 
and $k\mapsto p_+(-k)$ is the density of probability associated to Call option prices with implied volatility 
$k\mapsto w(k):=v(-k)$.
\end{enumerate}
\end{proposition}
The strict positivity of the function~$v$ in Assumption~(I) ensures that the support of the underlying distribution is the whole real line.
One could bypass this assumption by considering  finite support as in~\cite{RT10}.
In the latter---slightly more general---case, the statements and proofs would be very analogous but much more notationally inconvenient.
Symmetry properties of the implied volatility have been investigated in the literature, and we refer the interested reader
to~\cite{CarrLee, GuliSym, RT}.
This proposition has been intentionally stated in a maturity-free way: it is indeed a purely `marginal' or cross-sectional statement,
which does not depend on time. 
A natural question arises then: can such a function $v$, satisfying the assumptions of Proposition~\ref{prop:nomass},
represent the total implied variance smile at time~$1$ associated to some martingale 
(issued from~$1$ at time zero)? 
The answer is indeed positive and this can be proved as follows.
Consider the natural filtration $\mathbb{B}$ of a standard (one-dimensional) 
Brownian motion~$(B_t)_{t\geq0}$. 
Let~$P$ be the cumulative distribution function associated to $p_-$ characterised in Proposition~\ref{prop:nomass}, 
and~$\Nn $ the Gaussian cumulative distribution function. 
Then the random variable $X := P^{-1}(\Nn (B_1))$ has law~$P$, and~$E(X)=1$. 
Set now $M_s := E(X|\mathbb{B}_s)$, then $M$ is a martingale issued from~$1$. 
Note that $M$ is even a Brownian martingale and therefore a continuous martingale. 
The associated Call option prices $E[(M_s -K)_+]$ uniquely determine a total implied variance surface~$(t,k)\mapsto w(k,t)$ 
such that $v=w(1,\cdot)$.

\subsection{Volatility parameterisations}
In~\cite{Gatheral:2004}, Gatheral proposed a parameterisation for the implied volatility, the now famous SVI (`Stochastic Volatility Inspired').
However, finding necessary and sufficient conditions preventing static arbitrage have been inconclusive so far.
Recently, Gatheral and Jacquier~\cite{gj} extended this approach and introduced the following parameterisation for the total implied variance~$w$:
\begin{equation}\label{eq:SVIW}
w(k,t) \equiv \frac{\theta_t}{2}\left\{1+\rho k\varphi(\theta_t)+\sqrt{(k\varphi(\theta_t)+\rho)^2+(1-\rho^2)}\right\},
\end{equation}
with $\theta_t>0$ for $t>0$ and $\varphi$ is a smooth function from~$\RR_+^*$ to~$\RR_+$ 
and $\rho\in (-1,1)$.
The main result in their paper (Corollary 5.1) is the following theorem, 
which provides sufficient conditions for the implied volatility surface~$w$ to be free of static arbitrage:

\begin{theorem} \label{theogj}
The surface~\eqref{eq:SVIW} is `free of static arbitrage' if the following conditions are satisfied:
\begin{enumerate}
  \item $\partial_t \theta_t\geq 0$ for all $t>0$;
  \item $\varphi(\theta)+\theta\varphi'(\theta)\geq 0$ for all $\theta>0$;
  \item $\varphi'(\theta) < 0$ for all $\theta>0$;
  \item $\theta\varphi(\theta)(1+|\rho|) < 4$ for all $\theta>0$;
  \item $\theta\varphi(\theta)^2(1+|\rho|)\leq 4$ for all $\theta>0$.
\end{enumerate}
\end{theorem}
A few remarks are in order here:
\begin{enumerate}
\item the conditions in Theorem~\ref{theogj} are sufficient, but not necessary;
\item the full characterisation of the functions $\varphi$ guaranteeing absence of (static or not) arbitrage 
in the symmetric SVI case $\rho=0$ is left open;
\item it would be useful to `parameterise' the set of functions $\varphi$ satisfying the conditions of Theorem~\ref{theogj}. 
This could lead to easy-to-implement calibration
algorithms among the whole admissible class, without being tied to a particular family as in~\cite{gj}.
\end{enumerate}

In this paper, we try to settle all these points, and state our results in a more general framework, 
not tied to the specific shape of the SVI model, 
by considering implied volatility surfaces of the form
\begin{equation}\label{eq:GenSVI}
w(k,t)=\theta_t\Psi(k \varphi(\theta_t)),
\quad\text{for all }k\in\RR, t\geq 0,
\end{equation}
together with the following assumptions:
\begin{assumption}\label{assu:SVI}\text{}
\begin{enumerate}[(i)]
\item $\theta \in \Cc^1(\RR_+^*\to \RR_+^*)$, is not constant, 
$\lim_{t\downarrow 0} \theta_t = 0$, and $\theta_\infty:=\lim_{t\uparrow\infty}\theta_t$ is well defined in~$(0,\infty]$;
\item $\varphi \in \Cc^1(\RR_+^*\to \RR_+^*)$, 
and $\lim_{u\uparrow\infty}\varphi(u)$ 
is well defined in $(0,\infty]$;
\item $\Psi \in \Cc^2(\RR\to \RR_+^*)$
with $\Psi(0)=1$ and $\Psi$ is not constant;
\item for any $k\in\RR$, $\lim_{t\downarrow 0}w(k,t) =0$.
\end{enumerate}
\end{assumption}
The time-dependent function~$\theta$ models the at-the-money total variance; 
the assumption on its behaviour at the origin is thus natural. 
A constant function~$\Psi$ corresponds to deterministic time-dependent volatility, 
a trivial case we rule out here.
Likewise, were~$\theta$ assumed to be constant, it would be null everywhere, 
which we shall also not consider.
Assumption~(iv) ensures that at maturity, European Call option prices are equal to their payoffs. 
We can recast it in terms of assumptions on~$\varphi$ and~$\Psi$, for example:\\
\textbf{Assumption (iv')}: $\varphi(\theta)$ converges to a non-negative constant as $\theta\downarrow 0$.\\
Indeed (iv'), together with (iii), clearly implies (iv).
We shall present another alternative below with the help of the `asymptotic linear' property of~$\Psi$ 
(Definition~\ref{def:Linear} and Assumption~\ref{assu:payoffBis}).
Assumption~(iii) may look strong from a purely theoretical point of view, but is always satisfied in practice.
In Section~\ref{sec:nonsmooth} though, we partially relax it (Assumption~\ref{assu:SVIweak}) 
to allow for possible kinks.
The main goal here is to provide sufficient conditions on the triplet~$(\theta, \varphi, \Psi)$ that will guarantee absence of static arbitrage.
Note that the SVI parameterisation~\eqref{eq:SVIW}
corresponds to the case $\Psi(z)\equiv \frac{1}{2} (1+\rho z +\sqrt{z^2+2\rho z+1})$,
which clearly satisfies Assumption~\ref{assu:SVI}(iii).
In the sequel, we shall refer to this case as the SVI case.
The next sections provide necessary and sufficient conditions on $\theta, \varphi$ and $\Psi$ 
to prevent static arbitrage.

\section{Elimination of calendar spread arbitrage} \label{sec_cal}
We first concentrate on determining (necessary and sufficient) conditions on the triplet~$(\theta, \varphi, \Psi)$
to eliminate calendar spread arbitrage.

\subsection{The first coupling condition}
The quantity $\partial_t w(k,t)$ in Definition~\ref{def:WeakArb} is nothing else than the numerator of the local volatility expressed 
in terms of the implied volatility, i.e. Dupire's formula (see~\cite{jimbook}).
Define now the functions~$F:\RR\to\RR$ and~$f:\RR\to\RR$~by
\begin{equation}\label{eq:Ff}
F(z):=z\frac{\Psi'(z)}{\Psi(z)},\qquad
f(u):=u\frac{\varphi'(u)}{\varphi(u)}.
\end{equation}
They will play a major role in our analysis, and 
Assumption~\ref{assu:SVI}(iii) implies that $F(z) \sim \Psi'(0) z/\Psi(0)$ at the origin and $F(0)=0$.
Note that $\Psi$ and $\varphi$ can be recovered through the identities
$$
\Psi(z)=\exp\left({\int_0^z \frac{F(u)}{u}}\D u\right),
\qquad
\varphi(u)=\varphi(r)\exp\left({\int_{r}^{u}\frac{f(v)}{v}}\D v\right),
$$
for some arbitrary constant~$r>0$.
The following proposition gives new conditions for absence of calendar spread arbitrage.
\begin{proposition}[First coupling condition] \label{prop:first_coupling}
The surface~\eqref{eq:GenSVI} is free of calendar spread arbitrage if and only if the following two conditions hold:
\begin{enumerate}[(i)]
\item $\theta$ is non-decreasing;
\item $1+F(z)f(u)\geq 0$ for any $z\in\RR$ and $u\in(0, \theta_\infty)$.
\end{enumerate}
\end{proposition}
\begin{proof}
By Definition~\ref{def:WeakArb}, the surface defined by~\eqref{eq:GenSVI} is free of calendar spread arbitrage if and only if
\begin{equation} \label{calcond}
\partial_tw(k,t)=\theta_t'\Psi(z)+\theta_t\Psi'(z)k\varphi'(\theta_t)\theta_t'\geq 0,
\quad\text{for all }k\in\mathbb{R}, t>0,
\end{equation}
where $z:=k\varphi(\theta_t)$. 
Since $\Psi$ is strictly positive by Assumption~\ref{assu:SVI}(iii), the inequality~\eqref{calcond} is equivalent to
$\theta'_t\left(1+F(z)f(\theta_t)\right)\geq 0$ for all  $z\in\mathbb{R}$, $t>0$, 
with~$F$ and~$f$ defined in~\eqref{eq:Ff}.
For $k=0$ we get $\theta'_t\geq0$ for all $t>0$. 
Otherwise~(ii) is necessary and sufficient for the surface to be free of calendar spread arbitrage.
\end{proof}

\begin{remark}
We do not assume here that $\theta_\infty$ is infinite. 
In most popular stochastic volatility models with or without jumps, $\theta_\infty$ is infinite.
Rogers and Tehranchi~\cite{RT10} showed that for a non-negative martingale~$(S_t)_{t\geq 0}$ 
the equality $\theta_\infty=\infty$ is equivalent to the almost sure equality $\lim_{t\uparrow\infty}S_t=0$ (where the limit exists by the martingale convergence theorem).
However, it may occur that $\theta_\infty<\infty$. 
As a corollary of coupling properties of stochastic volatility models, Hobson~\cite{Hobson} provides instances where such a phenomenon appears, for example the SABR~\cite{Hagan} model with $\beta=1$.
\end{remark}

\begin{remark}\label{rem:F}
Condition~(ii) in Proposition~\ref{prop:first_coupling} can be stated in a more compact way:
$$1-\sup{F_+} \sup{f_-}\geq0
\qquad\text{and}\qquad
1-\sup{F_-} \sup{f_+}\geq0,$$
where $f_+:=\max(f,0)$ and $f_-:=\max(-f, 0)$.
\end{remark}

Motivated by the celebrated moment formula in~\cite{Lee} (see also Theorem~\ref{thm:LeeMoment}), 
which forces the function~$\Psi$ to be at most linear at (plus/minus) infinity, 
let us propose the following definition:
\begin{definition}\label{def:Linear}
The function~$\Psi$ is said to be asymptotically linear if 
$\lim\limits_{z\to\pm\infty}\Psi'(z)=:\alpha_\pm \in \RR\setminus\{0\}$.
\end{definition}

With this definition, we can replace Assumption~\ref{assu:SVI}(iv) by
\begin{assumption}\label{assu:payoffBis}
$\Psi$ is asymptotically linear and $\lim_{\theta\downarrow 0}\theta\varphi(\theta) = 0$.
\end{assumption}

We now obtain a necessary condition on the behaviour of the function $\varphi$ in~\eqref{eq:GenSVI}.
\begin{proposition} \label{prop:first_lin}
If $\Psi$ is asymptotically linear and if there is no calendar spread arbitrage, 
then the map $u \mapsto u \varphi(u)$ is non-decreasing on~$\RR_+$.
\end{proposition}
\begin{proof}
Using~\eqref{eq:Ff}, if $\Psi$ is asymptotically linear, then 
$\lim_{z\rightarrow \pm\infty} z\Psi'(z)/\Psi(z) = \lim_{z \rightarrow \pm\infty} F(z)=1$, 
so that absence of calendar spread arbitrage implies $1+f(u)\geq 0$ for any $u\in(0, \theta_\infty)$
by Proposition~\ref{prop:first_coupling}(ii).
Since $\varphi$ is a strictly positive function by Assumption~\ref{assu:SVI}(ii), 
the proposition follows from~\eqref{eq:Ff}.
\end{proof}

Note that if $\lim_{z\rightarrow \pm \infty} \Psi'(z) = 0$ then 
the limit of the function~$F$ at (plus or minus) infinity does not necessarily exist.
Whenever it does, since $z \mapsto \Psi(z)/z$ is decreasing as $z \rightarrow \pm\infty$, 
the limit can take any value in $(-\infty,1)$.

\subsection{Application to SVI}\label{sec:CalendarSVI}
In the SVI case~\eqref{eq:SVIW}, we have $\Psi'(z)\equiv\frac{1}{2}\left(\rho+\frac{z+\rho}{\sqrt{z^2+2\rho z+1}}\right)$ with $|\rho| < 1$,
so that $\Psi$ is asymptotically linear with $\alpha_+=\rho+1$ and $\alpha_-=\rho-1$.
Therefore Proposition~\ref{prop:first_lin} applies, 
and a necessary condition is that $u \mapsto u \varphi(u)$ is not decreasing.
In~\cite[Theorem 4.1]{gj}, this condition, together with $\varphi$ being non-increasing, are shown to be sufficient
to avoid calendar spread arbitrage.
In the case of the symmetric SVI model, the following corollary relates our conditions to those in~\cite{gj}.
\begin{corollary}
In the symmetric SVI case, the necessary condition of Proposition~\ref{prop:first_lin} is also sufficient.
\end{corollary}
\begin{proof}
In the symmetric case~$\rho=0$, we can compute explicitly
\begin{equation} \label{sym_dev}
\Psi(z) = \frac{1+\sqrt{1+z^2}}{2},\qquad
\Psi'(z) = \frac{1}{2}\frac{z}{\sqrt{1+z^2}}, \qquad
\Psi''(z) = \frac{1}{2(1+z^2)^{3/2}},
\qquad\text{for all }z\in\RR,
\end{equation}
and therefore 
$$
F(z)=\frac{z^2}{\sqrt{1+z^2}\left(1+\sqrt{1+z^2}\right)}
\qquad\text{and}\qquad
F'(z) = \frac{z}{(1+z^2)^{3/2}},
\qquad\text{for all }z\in\RR.
$$
It is then clear that the even function $F$ is strictly increasing on~$\RR_+^*$ 
and strictly decreasing on~$\RR_-^*$
with a global minimum attained at the origin for which $F(0)=0$.
In light of Remark~\ref{rem:F}, we have $\sup{F_+}=1$ and $\sup{F_-}=0$.
By Proposition \ref{prop:first_coupling} there is hence no calendar spread arbitrage if and only if $f(u)\geq-1$, which is equivalent to 
$u \mapsto u \varphi(u)$ being non-decreasing. 
\end{proof}

\section{Elimination of butterfly arbitrage \label{sec_butter}}
We now consider butterfly arbitrage which, probably not surprisingly, is more subtle to handle.
We first start with a general result (Section~\ref{sec:Coupling21}), which is unfortunately not that tractable in practice.
When the function $\Psi$ is asymptotically linear, however, more elegant formulations are available, 
and we provide necessary and sufficient conditions precluding static arbitrage (Section~\ref{sec:Linear}).
In the particular example of the symmetric SVI function (Section~\ref{sec:ApplicationSVI}), 
we put these results in action,
where everything is computable explicitly.
Finally, in Section~\ref{sec:nonsmooth}, we address a delicate issue, allowing for the possibility of non-smooth functions, 
thereby enlarging the class of arbitrage-free volatility surfaces.

\subsection{The second coupling condition}\label{sec:Coupling21}
We consider here the positivity condition $\Ll w (k,t)\geq0$ from Definition~\ref{def:WeakArb}, 
and reformulate the butterfly arbitrage condition in our setting.
We first start with a general formulation, and then consider the asymptotically linear case (for the function $\Psi$), 
which turns out to be more tractable.
For any~$u \in (0,\theta_\infty]$, define the set 
\begin{equation}\label{eq:SetZ+}
\Zz_+(u) := \left\{z \in \mathbb{R}: \frac{1}{4 u} \left(\frac{\Psi'(z)^2}{\Psi(z)}-2\Psi''(z)\right)
+\frac{\Psi'(z)^2}{16} > 0\right\},
\end{equation}
as well as the function
$\Lambda:\left\{(u, z): u\in (0, \theta_\infty], z \in \Zz_+(u)\right\}\to \RR \cup\{+\infty\}$ by
\begin{equation}\label{eq:Lambda}
\Lambda(u, z) := 
\left(\frac{1}{4 u} \left(\frac{\Psi'(z)^{2}}{\Psi(z)}-2\Psi''(z)\right)+\frac{\Psi'(z)^{2}}{16}\right)^{-1}\left(1-\frac{z\Psi'(z)}{2\Psi(z)}\right)^2.
\end{equation}
\begin{proposition}[Second coupling condition, general formulation]
\label{prop:Coupling2}
The surface~$w$ given in~\eqref{eq:GenSVI} is free of butterfly arbitrage if and only if
\begin{equation}\label{eq:Coupling2Iff}
(u \varphi(u))^2 \leq 
\inf_{z \in \Zz_+(u)} \Lambda(u, z),
\quad\text{for all }u \in (0, \theta_\infty).
\end{equation}
\end{proposition}

\begin{proof}
From~\eqref{eq:OpL} and~\eqref{eq:GenSVI}, we clearly have
$\partial_kw(k,t)=\theta_t\Psi'(z)\varphi(\theta_t)$, and 
$\partial^2_{kk}w(k,t)=\theta_t\Psi''(z)\varphi(\theta_t)^2$
for all $k\in\RR$ and $t>0$.
Therefore, with $z:=k\varphi(\theta_t)$, 
\begin{align}
\Ll w (k,t) & = \left(1-\frac{k\partial_kw(k,t)}{2w(k,t)}\right)^2-\frac{(\partial_k w(k,t))^2}{4}\left(\frac{1}{w(k,t)}+\frac{1}{4}\right)+\frac{\partial_{kk}^2w(k,t)}{2}\nonumber\\
 & = \left(1-\frac{k\theta_t\Psi'(z)\varphi(\theta_t)}{2\theta_t\Psi(z)}\right)^2-\frac{(\theta_t\Psi'(z)\varphi(\theta_t))^2}{4}\left(\frac{1}{\theta_t\Psi(z)}+\frac{1}{4}\right)+\frac{\theta_t\Psi''(z)\varphi(\theta_t)^2}{2}\nonumber\\
 & = \left(1-\frac{z\Psi'(z)}{2\Psi(z)}\right)^2-(\theta_t \varphi(\theta_t))^2 
\left\{\frac{1}{4 \theta_t} \left(\frac{(\Psi') ^{2}(z)}{\Psi(z)}-2\Psi''(z)\right)+\frac{(\Psi')^{2}(z)}{16}\right\},\label{ineq:condfly}
\end{align}
and the proposition follows from the definition of~$\Zz_+(u)$.
Indeed, on~$\RR\setminus \Zz_+(u)$, butterfly arbitrage is clearly precluded for any $u>0$, since both terms on the right-hand side of~\eqref{ineq:condfly} are non-negative.
\end{proof}

\subsection{The asymptotically linear case} \label{sec:Linear}
We now consider the case where $\Psi$ is asymptotically linear (Definition~\ref{def:Linear}). 
Define the sets
\begin{equation}\label{eq:SetZBarP}
\overline{\Zz}_+ := \left\{z \in \mathbb{R}: \left(\frac{\Psi' (z)^{2}}{\Psi(z)}-2\Psi''(z)\right) > 0\right\},
\qquad
\overline{\Zz}_-:=\mathbb{R}\setminus \overline{\Zz}_+,
\quad\text{and}\quad
\aleph:= \{z\in\RR: \Psi'(z) = 0\},
\end{equation}
together with the complement in~$\RR$: $\aleph^c := \RR\setminus\aleph$, as well as the,
possibly infinite, quantity
\begin{equation}\label{eq:Minf}
M_\infty:=\lim_{u\uparrow \theta_\infty}u \varphi(u).
\end{equation}

The following proposition, proved in Appendix~\ref{sec:prop:noCSAStrictly}, is a reformulation of Proposition~\ref{prop:Coupling2} 
in the asymptotically linear case, and provides sufficient and necessary conditions for the surface~\eqref{eq:GenSVI} to be free of butterfly arbitrage.
\begin{proposition} \label{prop:noCSAStrictly}
Assume that $\Psi$ is asymptotically linear and there is no calendar spread arbitrage. 
Then~$\overline{\Zz}_+$ is neither empty nor bounded from above.
Moreover, there is no butterfly arbitrage if and only if the following two conditions hold
(recall that the functions~$\Zz_+$ and~$\Lambda$ are defined in~\eqref{eq:SetZ+} and~\eqref{eq:Lambda}):
\begin{enumerate}[(i)]
\item 
\begin{equation*}
\begin{array}{lll}
M_\infty^2 & \displaystyle \leq \inf_{z \in \overline{\Zz}_- \cap \Zz_+(\theta_\infty)\cap\aleph^c}
\Lambda(\theta_\infty, z),
 & \text{if }\theta_\infty<\infty,\\
\\
M_\infty & \displaystyle \leq \inf_{z\in \overline{\Zz}_-\cap\aleph^c}\left|\frac{4}{\Psi'(z)}-\frac{2z}{\Psi(z)}\right|,
& \text{otherwise};
\end{array}
\end{equation*}
\item for any $u \in (0, \theta_\infty)$,
$\displaystyle (u \varphi(u))^2 \leq \inf_{z \in \overline{\Zz}_+} \Lambda(u, z)$.
\end{enumerate}
\end{proposition}
\begin{remark}
Case (ii) actually includes two cases:
$\overline{\Zz}_+\cap\aleph^c$ and $\overline{\Zz}_+\cap\aleph$.
On the former, the function $\Lambda(u,\cdot)$ is well defined and the infimum can be searched for without any confusion.
On $\overline{\Zz}_+\cap\aleph$, however, the function $z\mapsto\Lambda(u,z)$ reduces to $-2u/\Psi''(z)$,
which is always strictly positive.
Note further that, from~\eqref{ineq:condfly}, if $\Psi'(z) = \Psi''(z) = 0$, then positivity of $\Ll w(k,t)$
is automatically guaranteed.
\end{remark}
The following corollary is an immediate consequence of this proposition, in the case $\theta_\infty=\infty$.
\begin{corollary}\label{cor:EasyCondition1}
If $\Psi$ is asymptotically linear and $\theta_\infty=\infty$, then (allowing infinity)
$$
M_\infty \leq \inf_{z\in\mathbb{R}}\left|\frac{4}{\Psi'(z)}-\frac{2z}{\Psi(z)}\right|$$
is a necessary condition for absence of butterfly arbitrage.
In particular $M_\infty \leq 2/\sup\{|\alpha_+|, |\alpha_-|\}$.
\end{corollary}

A little work on the proposition above yields the following sufficient condition preventing butterfly arbitrage, which is easier to check in practice.
\begin{corollary}
Assume that~$\Psi$ is asymptotically linear, that there is no calendar spread arbitrage 
and that $\aleph = \emptyset$.
Assume further that for any $u\in (0,\theta_\infty)$, the inequality in Proposition~\ref{prop:noCSAStrictly}(ii) is strict.
Then the corresponding implied volatility surface is free of static arbitrage.
\end{corollary}

\begin{proof}
In our setting ($\Psi$ asymptotically linear),
$\lim_{k \uparrow \infty} \frac{w(k,t)}{k} =  \theta_t \varphi(\theta_t) \alpha_+$,
so that we only need to prove that $\theta_t \varphi(\theta_t) < \frac{2}{\alpha_+}$, 
since $\lim_{k\uparrow \infty}\frac{w(k,t)}{k}<2$ clearly implies the LMB condition.
For any $z \in \overline{\Zz}_+$ (defined in~\eqref{eq:SetZBarP}), note that
$$
\Lambda(\theta_t, z) = \frac{\left(1-\frac{z\Psi'(z)}{2\Psi(z)}\right)^2}
{\frac{1}{4 \theta_t} \left(\frac{\Psi'(z)^{2}}{\Psi(z)}-2\Psi''(z)\right)+\frac{\Psi'(z)^{2}}{16}}
\leq \frac{\left(1-\frac{z\Psi'(z)}{2\Psi(z)}\right)^2}{\frac{\Psi'(z)^{2}}{16}}.
$$
Applying this to a sequence in~$\overline{\Zz}_+$ diverging to infinity yields
$(\theta_t \varphi(\theta_t) )^2 < \frac{4}{\alpha_+^2}$
and the result follows.
\end{proof}

\subsection{Application to symmetric SVI}\label{sec:ApplicationSVI}
As in Section~\ref{sec:CalendarSVI} above, we show that in the symmetric SVI case ($\rho=0$), 
all our expressions above are easily computed and give rise to simple formulations.
It is clear that the set~$\aleph$ defined in~\eqref{eq:SetZBarP} is empty in this case.
Let us define the functions $A$, $Y$ and $A^*$ by
\begin{equation}\label{eq:AYStar}
\begin{array}{rl}
A(y, u) & := \displaystyle \frac{16 uy(y+1)}{8(y-2) + uy(y-1)},
\qquad
A^*(u) := A(Y(u), u),\\
Y(u) & := \displaystyle \frac{2}{1-u/4} +\sqrt{\left (\frac{2}{1-u/4}\right)^2 + \frac{2}{1-u/4}}.
\end{array}
\end{equation}
Of course we only define these functions on their effective domains, the forms of which we omit for clarity.
The following proposition makes the conditions of Proposition~\ref{prop:noCSAStrictly} explicit in the symmetric SVI case.
\begin{proposition}
In the symmetric SVI~\eqref{eq:SVIW} case $\rho=0$, there is no butterfly arbitrage if and only if
$$
\left(u \varphi(u)\right)^2 \leq A^*(u) \ind_{\{u<4\}}+ 16\, \ind_{\{u\geq4\}},
\qquad\text{for all }u \in (0, \theta_\infty).$$
\end{proposition}
\begin{proof}
Define $y_z:= \sqrt{1+z^2}$; then
$$
\left(\Psi'(z)\right)^2-2 \Psi(z) \Psi''(z) = \frac{1}{4} \frac{(y_z-2)(y_z+1)^2}{y_z^3},
\qquad
1-\frac{z \Psi'(z)}{2 \Psi(z)} = \frac{1}{2}\left(1+\frac{1}{y_z}\right),\\
\qquad
\Psi'(z)^2 = \frac{1}{4}\left(1 - \frac{1}{y_z^2}\right).
$$
Since $\Psi(z)>0$ for all $z\in\mathbb{R}$, the first equation implies that~$\overline{\Zz}_+$ defined in~\eqref{eq:SetZBarP} is equal to~$\RR\setminus [-\sqrt{3},\sqrt{3}]$. 
For any fixed $u$, the function appearing on the right-hand side of Proposition~\ref{prop:noCSAStrictly}(ii) 
simplifies to~$A(y, u)$ given in~\eqref{eq:AYStar}.
In particular
$A(2,u)=48$ and $\lim_{y\uparrow\infty}A(y,u)=16$.
For any $u\geq 0$, we have
$$
\partial_y A(y,u) = \frac{128 u B_u(y)}{\left(8y-16+y^2 u-yu\right)^2},
$$
where $B_u(y):=\left(1-\frac{u}{4}\right) y^2 -4 y-2$.
When $u\geq 4$, $B_u$ is concave on $(2,\infty)$ with $B_u(2)=-(6+u)<0$, and hence
the map $y\mapsto A(y,u)$ is decreasing on $(2,\infty)$ and its infimum is equal to
$\lim_{y\uparrow\infty}A(y,u)=16$.
For $u\in[0,4)$, the strict convexity of $B_u$ and the inequality $B_u(2) = -(6+u)<0$ implies that 
the equation $B_u(y)=0$ has a unique solution in~$(2,\infty)$, which in fact is equal to $Y(u)$ given in~\eqref{eq:AYStar}.
Then the map $y\mapsto A(y,u)$ is decreasing on $(2,Y(u))$ and increasing on $(Y(u),\infty)$.
Its infimum is attained at~$Y(u)$ and is equal to $A^*(u)$ defined in~\eqref{eq:AYStar}.
\end{proof}

\begin{remark}\label{rem:Comparison}
In~\cite{gj}, the authors prove that the two conditions (altogether) $u\varphi(u)<4$ and 
$u\varphi(u)^2<4$ (for all~$u\geq 0$) are sufficient to prevent butterfly arbitrage in the uncorrelated ($\rho=0$) case.
These two conditions can be combined to obtain
$(u\varphi(u))^2<16\min(1,\varphi(u)^{-2})$.
A tedious yet straightforward computation shows that $A^*$ is increasing on $[0,4)$
and maps this interval to $[0,16)$.
Notwithstanding the fact that our condition is necessary and sufficient, it is then clear
that 
\begin{enumerate}[(i)]
\item for $u\geq 4$, it is also weaker than the one in~\cite{gj} whenever $\varphi(u)<1$;
\item for $u<4$ (which accounts for most practically relevant cases) it is weaker
whenever~$16/\varphi(u)<A^*(u)$.
\end{enumerate}
In particular, item~(ii) could be used as a sufficient and necessary lower bound condition (depending on~$u$)
for the function $\varphi$ on $[0,4)$.
\end{remark}

\subsection{Non-smooth implied volatilities}\label{sec:nonsmooth}
The formulation of arbitrage freeness in~\cite[Theorem 2.1]{roper} 
is minimal in the sense that the regularity conditions on the Call option prices are necessary and sufficient: 
to be convex in the strike direction and non-decreasing in the maturity direction.
The implied volatility formulation (\cite[Theorem 2.9, condition IV.1]{roper} and Theorem~\ref{thm:RoperNoArb}(i) above)
however, assumes that the total variance is twice differentiable in the strike direction. 
This regularity is certainly not required; 
in fact,  the author~\cite[Theorem 2.9]{roper} proves the latter by checking the necessary assumptions 
on the behaviour of the Call price (\cite[Theorem 2.1]{roper}) defined by 
$\BS(\E^k, w(k,t))$,
with $\BS$ defined in~\eqref{eq:BlackScholes}.
More precisely, Roper uses the regularity assumption in~$k$ of~$w$ in order to define pointwise the second derivative of this Call price function with respect to the strike.
He then proves that the latter is positive, henceforth obtaining the convexity of the price with respect to the strike~\cite[Theorem 2.1, Assumption A.1]{roper}.
It turns out that the same result can be obtained without this regularity assumption.
Let $\widetilde{L}^\infty_{+}(\RR\to\RR_+^*)$ denote the 
space of strictly positive, continuous, functions on the real line, 
differentiable except possibly at finitely many points, 
and with derivatives in~$L^\infty_{\loc}(\RR\to\RR)$, 
the space of locally essentially bounded measurable functions.
Introduce then the functional~$\Mm$ on~$\widetilde{L}^\infty_{+}(\RR\to\RR_+^*)$ by
\begin{equation}\label{eq:M}
\Mm_v(k) := \left(1-\frac{kv'(k)}{2v(k)}\right)^2 - \frac{v'(k)^2}{4}\left(\frac{1}{v(k)}+\frac{1}{4}\right),
\qquad\text{for all }k\in\RR.
\end{equation}

\begin{proposition}
For any~$v \in \widetilde{L}^\infty_{+}(\RR\to\RR_+^*)$, the following hold:
\begin{enumerate}
\item the functional~$\Mm_v$ in~\eqref{eq:M}
is well defined in $\widetilde{L}^\infty_{+}(\RR\to\RR_+^*)$, 
hence in the sense of distributions;
\item let~$v''$ denote the second derivative of~$v$ in the sense of distributions. 
Then the map $K \mapsto \BS(K, v(\log(K))$ is convex if and only if 
$\Ll v := \displaystyle \Mm_v + \frac{1}{2}v''$ is a positive distribution.
\end{enumerate}
\end{proposition}

We abuse the notation slightly by considering the same symbol for the operator~$\Ll$ 
here and in~\eqref{eq:OpL}, although they do not act on the same spaces;
this should however not create any confusion.

\begin{proof}
The first statement follows from the fact that $v$ is positive continuous and 
$v' \in L^\infty_{\loc}(\RR\to\RR)$.
Consider now a strictly positive smooth function~$\zeta$ with compact support, which integrates to one,
and regularise~$v$ by convolution as
$v_\eps(k) \equiv \eps^{-1}\zeta(k/\eps) \ast v(k)$.
Then~$v_\eps$ is a smooth strictly positive function, and Roper's computation~\cite[Theorem 2.9]{roper}
applies:
$$
\frac{\D^2 \BS(K, v_\eps(\log(K))}{\D K^2}
 = \frac{2\partial_w \BS(K, v_\eps(\log(K))}{K^2} \Ll v_\eps(\log(K)),
$$ 
where~$\Ll v_\eps$ is defined pointwise, 
and where $\partial_w \BS$ denotes the derivative of the function~$\BS$ 
with respect to its second component.
It follows that for any~$\phi \in \Cc^\infty(\RR_+)$ with compact support on~$\RR_+$,
\begin{align}\label{eq:IntegrationByParts}
\int_{\RR_+} \phi''(K) \BS(K, v_\eps(\log(K)) \D K
 & = \int_{\RR_+} \phi(K) \frac{\D^2 \BS(K, v_\eps(\log(K))}{\D K^2} \D K\\
 & =
2\int_{\RR_+}  \phi(K)\frac{\partial_w \BS(K, v_\eps(\log(K))}{K^2} \Ll v_\eps(\log(K))\D K,\nonumber
\end{align}
where the boundary terms cancel since~$\phi$ has compact support.
Mapping~$K\mapsto \E^{k}$, the last integral reads
$$
\int_{\RR} \phi(\E^k) \E^{-k} \partial_{w}\BS\left(\E^k, v_\eps(k)\right) \Ll v_\eps(k) \D k.
$$
When~$\eps$ tends to zero, 
$v_\eps$ converges pointwise to~$v$, 
$v'_\eps$ to~$v'$ almost everywhere, 
and~$v''_\eps$ to~$v''$ in the sense of distribution.
It follows that the map $\partial_{w}\BS\left(\E^\cdot, v_\eps(\cdot)\right) \Ll v_\eps(\cdot)$ 
converges to $\partial_{w}\BS\left(\E^\cdot, v(\cdot)\right) \Ll v(\cdot)$ 
in the sense of distribution (on~$\RR$). 
Now the first line of~\eqref{eq:IntegrationByParts} converges to 
$$
\int_{\RR_+} \phi''(K) \BS(K, v(\log(K))\D K = \langle\phi, \Pp\rangle_{\RR_{+}},
$$
where~$\Pp$ is the second derivative of $K \mapsto \BS(K, v(\log(K))$ in the sense of distribution,
and $\langle\cdot, \cdot\rangle_{\RR_+}$ the duality bracket.
Therefore,
$\langle \phi, \Pp\rangle_{\RR_+}
 = 2\langle \phi(\E^\cdot)\E^{-\cdot}, \partial_{w} \BS(\E^\cdot, v(\cdot)) \Ll v\rangle_{\RR}$,
so that~$\Pp$ is a positive distribution on~$\RR_{+}$ 
if and only if $\partial_{w}\BS(\E^\cdot, v(\cdot)) \Ll v$ is a positive distribution on~$\RR$.
Finally, the function $K \mapsto \BS(K, v(\log(K))$ is convex if and only if~$\Pp$ is a positive distribution;
since $\partial_{w}\BS(\E^{\cdot}, v(\cdot))$ is positive continuous, 
$\partial_{w}\BS(\E^\cdot, v(\cdot)) \Ll v$ is a positive distribution if and only if~$\Ll v$ is, 
which concludes the proof.

\end{proof}
Let us finally note that our assumptions on~$w$ are indeed minimal: 
conversely, if we start from an option price convex in~$K$,
its first derivative is defined almost everywhere, and so is that of~$w$ (in~$K$ or~$k$) 
since the Black-Scholes mapping in total variance is smooth.
Assumption~\ref{assu:SVI} imposes some (mild yet sometimes unrealistic) conditions on the volatility surface.
It turns out that our results are still valid under weaker conditions on the function~$\Psi$.
Recall first the following definition:
\begin{definition}\label{def}
A continuous function~$f$ is said to be of class $\mathcal{D}(\RR\to\RR)$ 
if there exist $a_0<a_1<\cdots<a_N$ (for some $N\in\mathbb{N}$),
such that $f \in \mathcal{C}^2(\RR \setminus \{a_0,\ldots,a_N\}\to\RR)$,
and such that the right and left limits
$\lim\limits_{a\downarrow a_i}f'(a)$ and~$\lim\limits_{a\uparrow a_i}f'(a)$ exist for each $i\in\{0,\ldots,N\}$.
\end{definition}
Consider now the following alternative to Assumption~\ref{assu:SVI}:
\begin{assumption}\label{assu:SVIweak}
Assumption~\ref{assu:SVI}(i), (ii) and (iv) are unchanged, but~(iii) is replaced by the weaker version:
$\Psi \in \mathcal{D}(\RR \to \RR_+^*)$, with $\Psi(0)=1$, $\Psi$ not constant.
\end{assumption}

Let $\mathcal{A}_{\Psi}$ denote the (possibly empty) set of discontinuity of $\Psi'$. 
Under our assumption, $\Psi''$ in the distribution sense is defined as a sum of a continuous measure 
on~$\RR\setminus \mathcal{A}_{\Psi}$ and of Dirac masses $\alpha_i \delta_i$ at each point of discontinuity 
$a_i \in \mathcal{A}_{\Psi}$.
We extend the results of Section~\ref{sec_butter} in the following way:
recall the sets $\Zz_+(\cdot)$ and $\overline{\Zz}_\pm$ in~\eqref{eq:SetZBarP} and~\eqref{eq:SetZ+}, 
and define
\begin{equation}\label{eq:SetZweak+}
\widetilde{\Zz}_+(u) := \Zz_+(u)\setminus\mathcal{A}_{\Psi}
\qquad\text{and}\qquad
\widetilde{\overline{\Zz}}_+:=\overline{\Zz}_+\setminus\mathcal{A}_{\Psi}.
\end{equation}
\begin{proposition}[Second coupling condition, general formulation]\label{prop:weak}
The surface~\eqref{eq:GenSVI} is free of butterfly arbitrage if and only if
the jumps of $\Psi'$ are non-negative and~\eqref{eq:Coupling2Iff} holds with~$\widetilde{\Zz}_+$ instead of~$\Zz_+$.
\end{proposition}

\begin{proof}
Similarly to the proof of~\ref{prop:Coupling2}, the continuous part of~$\Ll w$ has a density given by
$$
\left(1-\frac{z\Psi'(z)}{2\Psi(z)}\right)^2-(\theta_t \varphi(\theta_t))^2 
\left\{\frac{1}{4 \theta_t} \left(\frac{(\Psi') ^{2}(z)}{\Psi(z)}-2\Psi''(z)\right)+\frac{(\Psi')^{2}(z)}{16}\right\},
$$
for any $t>0$ and $k\in \RR\setminus\mathcal{A}_{\Psi}$,
and the first part of the proposition follows.
The remaining part of the distribution $\Ll w (k,t) $ is the sum of the disjoint Dirac masses 
$(\theta_t \varphi(\theta_t))^2 \alpha_i \delta_i $.
By localisation it is clear that the distribution $\Ll w (k,t) $ is positive if and only if its continuous part on 
$\RR \setminus \mathcal{A}_{\Psi}$ is positive and each
of its point mass distribution is positive. 
Since~$\alpha_i$ is non-negative if and only if~$\Psi'$ has a non-negative jump at~$a_i$, 
the rest of the proposition follows.
\end{proof}

Likewise, the analogue of Proposition~\ref{prop:noCSAStrictly} holds as follows:

\begin{proposition} \label{prop:noCSAStrictlyweak}
If $\Psi$ is asymptotically linear and if there is no calendar spread arbitrage, 
then~$\widetilde{\overline{\Zz}}_+$ is neither empty nor bounded from above.
Moreover, there is no butterfly arbitrage if and only if 
the jumps of $\Psi'$ are non-negative and 
Proposition~\ref{prop:noCSAStrictly}(i)-(ii) hold with $\widetilde{\Zz}_+(\cdot)$ and 
$\widetilde{\overline{\Zz}}_\pm$ instead of 
$\Zz_+(\cdot)$ and $\overline{\Zz}_\pm$.
\end{proposition}


\section{The quest for a non-SVI $\Psi$ function\label{sec_example}}
In order to find examples of pairs $(\varphi, \Psi)$, with $\Psi$ different from the SVI parameterisation~\eqref{eq:SVIW}, 
observe first that the second coupling condition (Proposition~\ref{prop:Coupling2}) is more geared 
towards finding out $\varphi$ given $\Psi$ than the other way round. 
We first start with a partial result (proved in Appendix~\ref{sec:prop:EffBoundPsi}) in the other direction, 
assuming that $\Psi$ is asymptotically linear.

\begin{proposition}\label{prop:EffBoundPsi}
If the generalised SVI surface~\eqref{eq:GenSVI} is free of static arbitrage, 
$\Psi$ is asymptotically linear and $\theta_\infty = \infty$, then there exist $z_+\geq 0$ and $\kappa\geq 0$ 
such that for all $z\geq z_+$ the following upper bound holds (with~$M_\infty$ defined in~\eqref{eq:Minf}):
$$
\Psi(z) \leq \kappa^2+\frac{2z}{M_\infty} - \kappa \sqrt{\kappa^2+\frac{2z}{M_\infty}}.
$$
\end{proposition}

Using this proposition, we now move on to specific examples of non-SVI families.

\subsection{First example of non-SVI function}\label{sec:nonSVI}
We here provide a triplet $(\theta,\varphi,\Psi)$, different from the SVI form~\eqref{eq:SVIW}, 
which characterises an arbitrage-free volatility surface via~\eqref{eq:GenSVI}. 
Let~$\theta_t\equiv t$ and 
$$
\varphi(u) :=
\left\{
\begin{array}{ll}
\displaystyle \frac{1-\E^{-u}}{u}, & \text{ if $u>0$},\\ 
1, & \text{ if $u=0$},
\end{array}
\right.
\qquad\text{and}\qquad
\Psi(z) := 
\displaystyle |z| + \frac{1}{2}\left(1+\sqrt{1+|z|}\right), 
\quad\text{for all }z\in\RR.
$$
A few remarks are in order:
\begin{itemize}
\item the function~$\varphi$ is continuous on~$\RR_+$;
\item $\theta_\infty=\infty$;
\item the map $u \mapsto u\varphi(u)$ is increasing and its limit is $M_\infty=1$;
\item the function $\Psi$---directly inspired from the computations in Proposition~\ref{prop:EffBoundPsi}---is symmetric and continuous on $\RR$.
It is also $\mathcal{C}^\infty$ on $\RR\setminus\{0\}$, and asymptotically linear. 
Its derivative is therefore $\mathcal{C}^1$ piecewise and has a positive jump at the origin, 
so that Propositions~\ref{prop:weak} and~\ref{prop:noCSAStrictlyweak} apply.
\end{itemize}
With these functions, the total implied variance~\eqref{eq:GenSVI} reads
$$
w(k,t) = k\left(1-\E^{-t}\right)+\frac{\sqrt{t}}{2}\left(\sqrt{t}+\sqrt{k\left(1-\E^{-t}\right)+t}\right),
\qquad\text{for all }k\in\RR,t\geq 0,
$$
and the following proposition (proved in Appendix~\ref{sec:prop:nonSVI1}) is the main result here:
\begin{proposition}\label{prop:nonSVI1}
The surface~$w$ is free of static arbitrage.
\end{proposition}

  \begin{figure}[htb!]
  \begin{center}
  \subfigure{\includegraphics[scale=0.3]{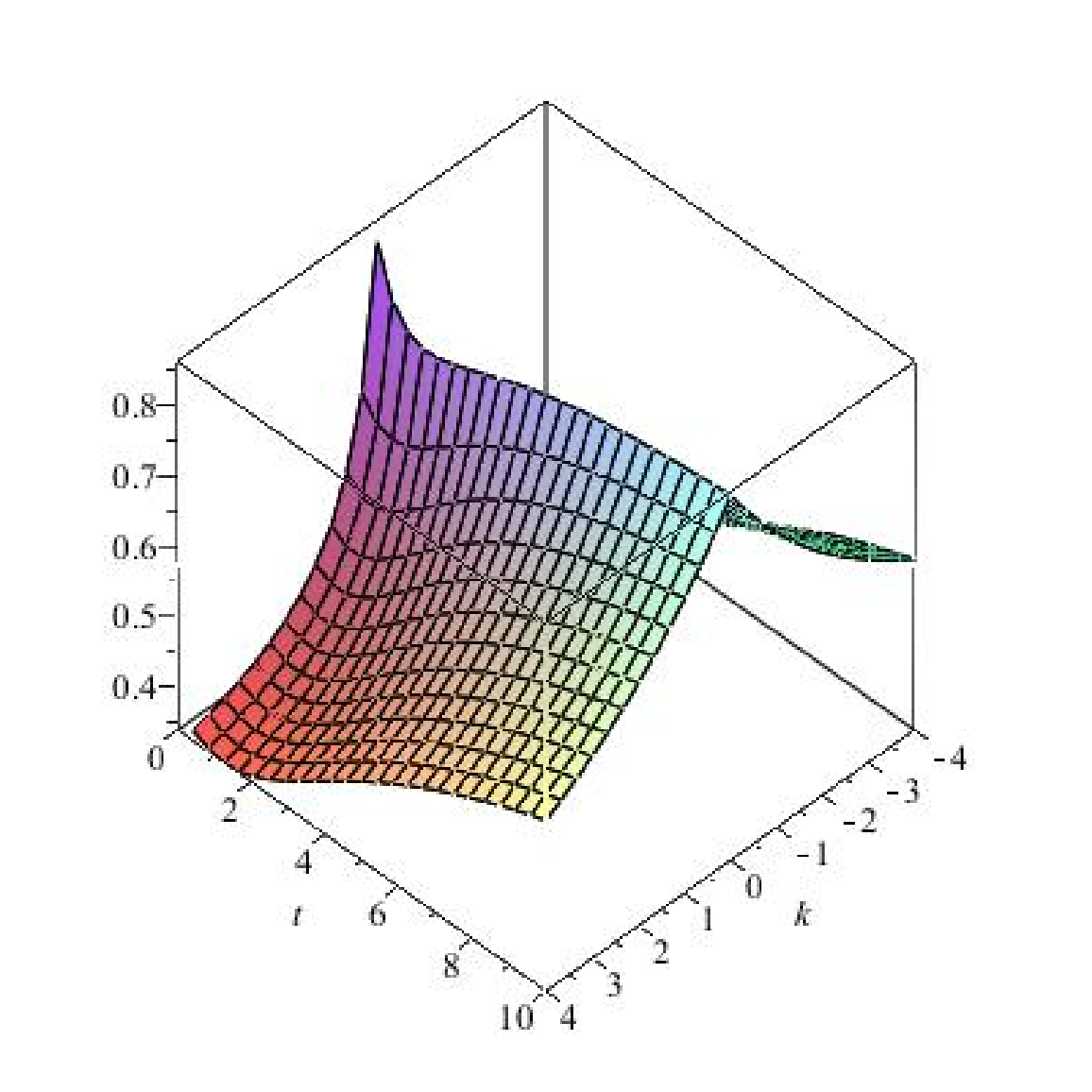}}
  \subfigure{\includegraphics[scale=0.25]{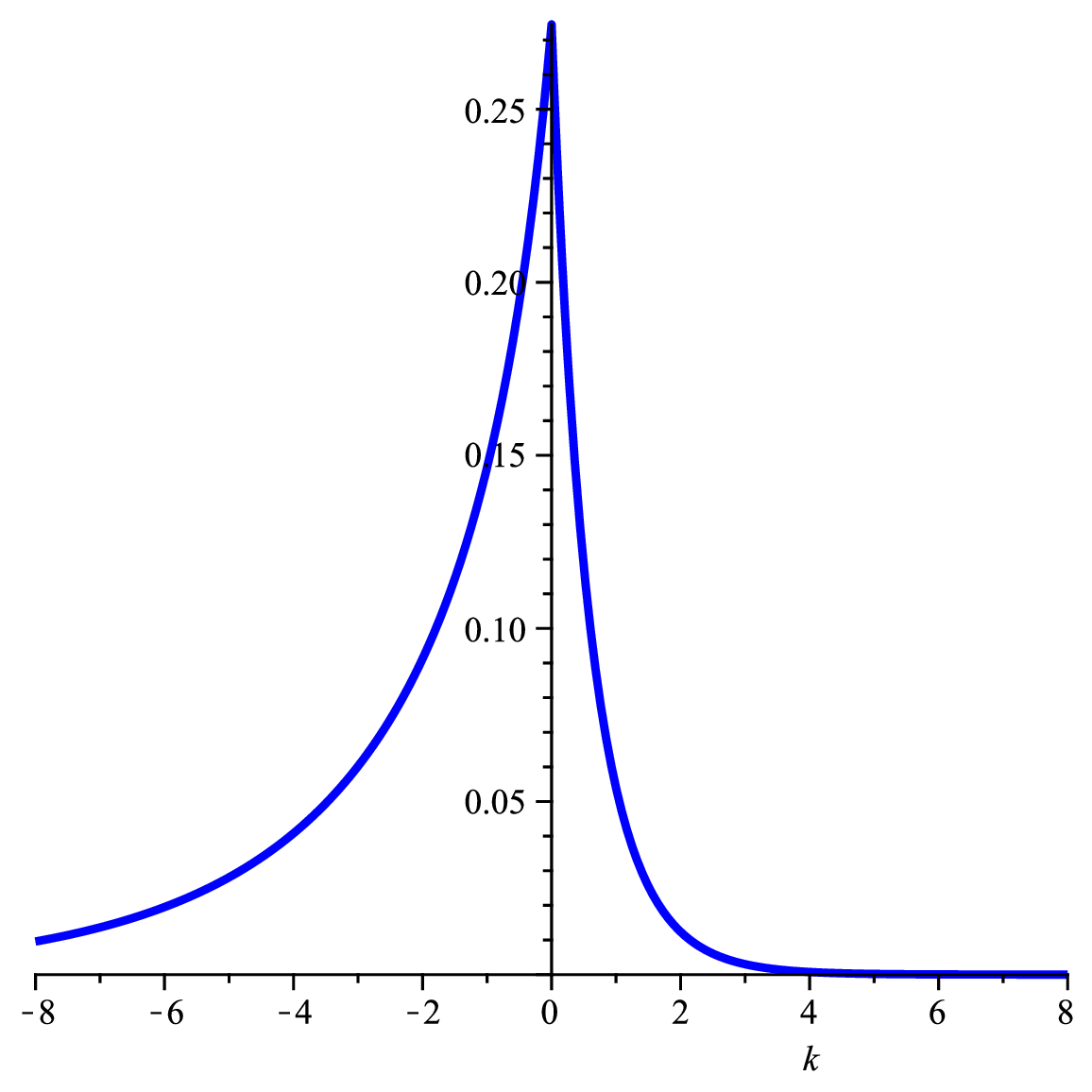}}
  \caption{Plot of the map $(k,t)\mapsto\mathcal{L}w(k,t)$ (left) 
  in the non-SVI case of Section~\ref{sec:nonSVI}, 
  and of the density at time $t=1$ (right).
}
  \end{center}
  \end{figure}

\subsection{Second example of non-SVI function}\label{sec:nonSVI2}
We propose a new triplet $(\theta,\varphi,\Psi)$ characterising an arbitrage-free volatility surface via~\eqref{eq:GenSVI}. 
Let~$\theta_t\equiv t$ and 
$$
\varphi(u) :=
\left\{
\begin{array}{ll}
\displaystyle \alpha\frac{1-\E^{-u}}{u}, & \text{ if $u>0$},\\ 
\alpha, & \text{ if $u=0$},
\end{array}
\right.
\qquad\text{and}\qquad
\Psi_\nu(z) := 
\displaystyle \left(1+|z|^\nu\right)^{1/\nu},
\quad\text{for }z\in\RR, 
$$
where $\nu \in (1,\infty)$ and $\alpha\in (0,\overline{\alpha})$ with $\overline{\alpha}\approx 1.33$.
Note that when $\nu=2$, modulo a constant, the function $\Psi_2$ corresponds to SVI.
We could in principle let $\alpha$ depend on $\nu$. 
The reason for the construction above is that 
we want to show that the corresponding implied volatility surface is free of static arbitrage 
for all $\nu>1$.
The same remarks as in the example in Section~\ref{sec:nonSVI} hold:
$\varphi$ is continuous on~$\RR_+$, $\theta_\infty=\infty$, 
$u \mapsto u\varphi(u)$ is increasing to $M_\infty=\alpha$
and $\Psi_\nu$ is symmetric and continuous on $\RR$. 
It is also $\mathcal{C}^\infty$ on $\RR\setminus\{0\}$,  $\mathcal{C}^1$ on $\RR$, and asymptotically linear.
The derivative $\Psi'$ has a positive jump at 0, so that we are back in the framework of Propositions~\ref{prop:weak} and~\ref{prop:noCSAStrictlyweak}.
With these functions, the total implied variance~\eqref{eq:GenSVI} reads
$$
w(k,t) = \theta_t\left(1+\frac{(1-\E^{-\theta_t})^\nu}{\theta_t^\nu}\alpha^\nu|k|^\nu\right)^{1/\nu},
\quad\text{for all }k\in\RR, t> 0,
$$
and we can check all the conditions preventing arbitrage (the proof is postponed to Appendix~\ref{sec:prop:nonSVI2}):
\begin{proposition}\label{prop:nonSVI2}
The surface~$w$ is free of static arbitrage.
\end{proposition}

  \begin{figure}[h]
  \begin{center}
  \subfigure{\includegraphics[scale=0.3]{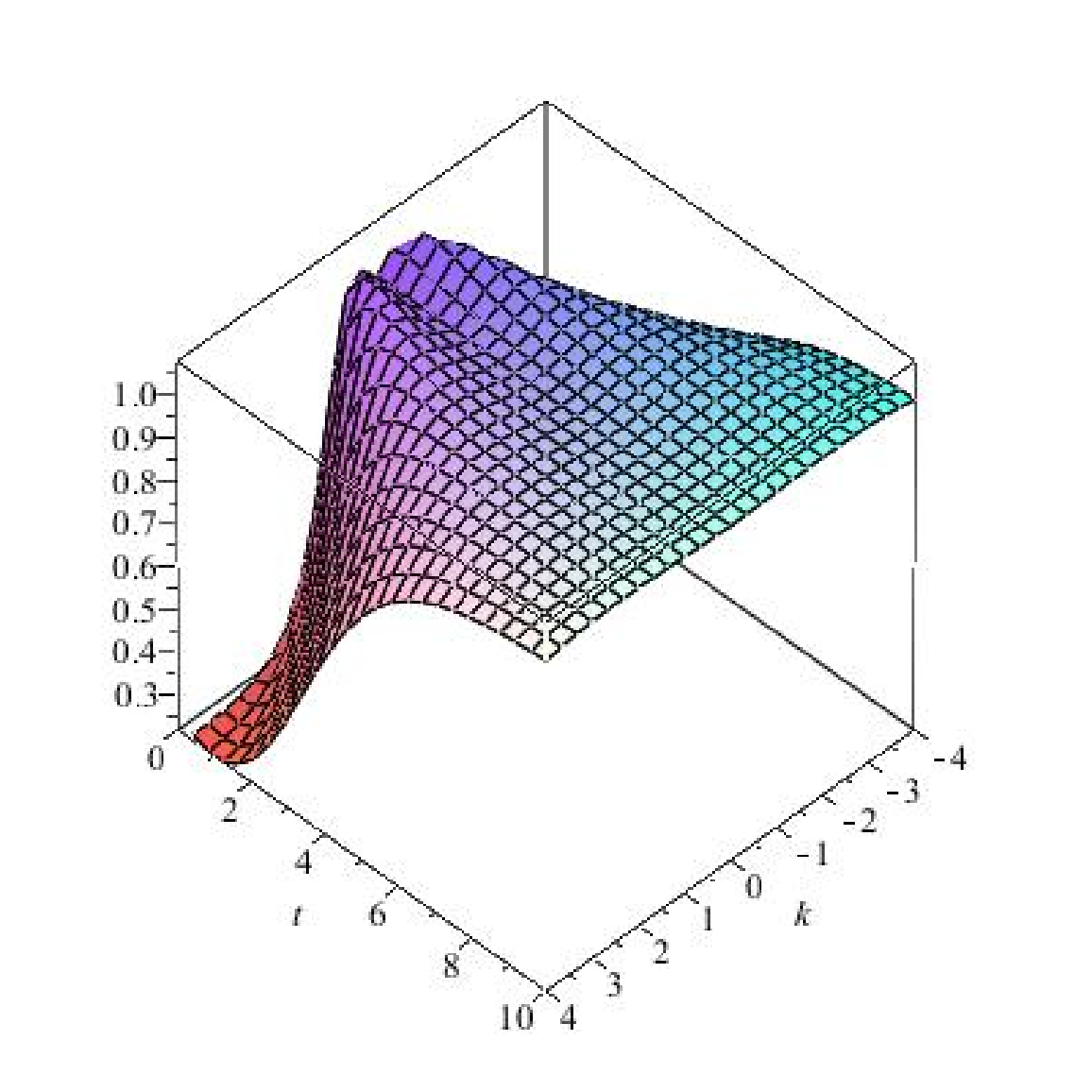}}
  \subfigure{\includegraphics[scale=0.25]{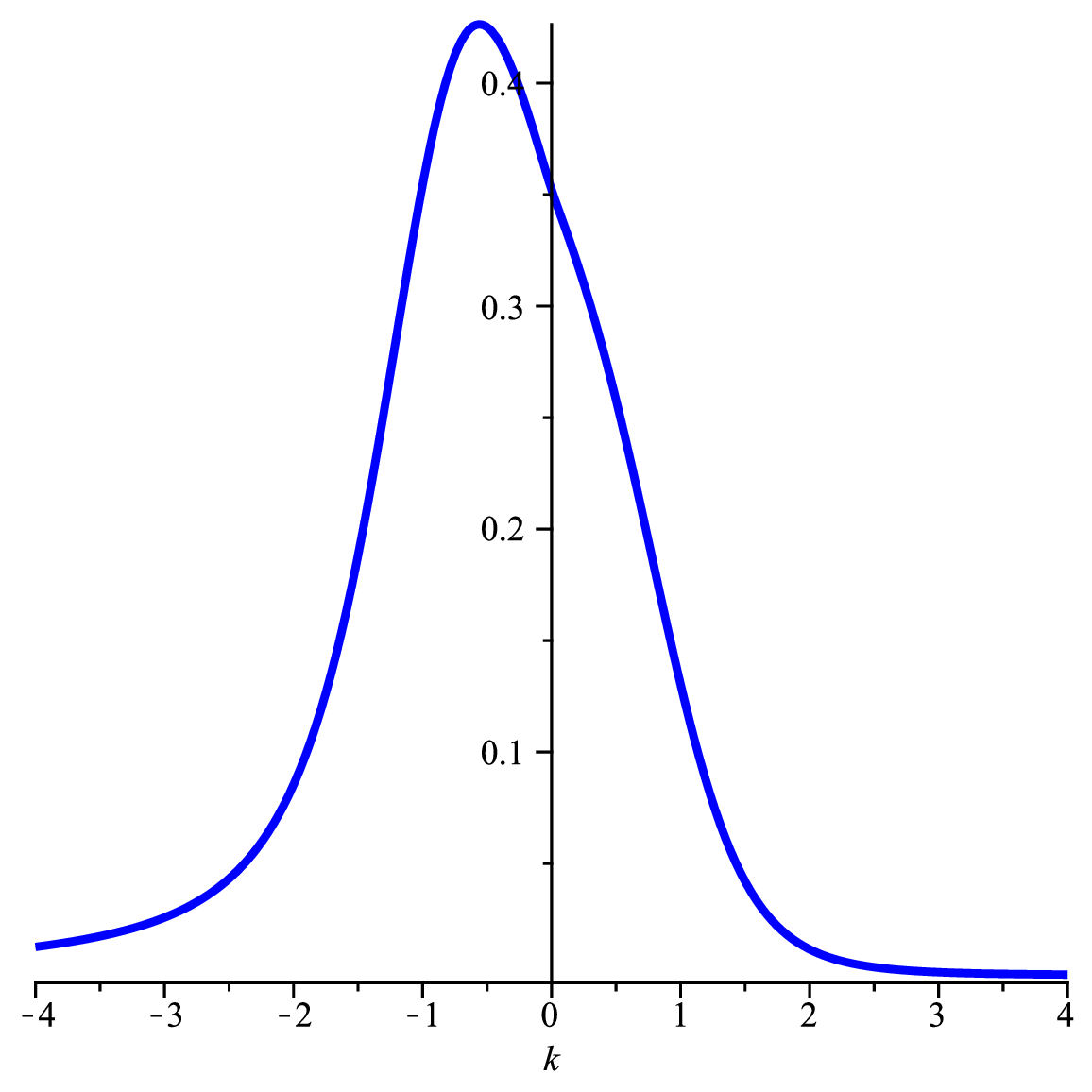}}
  \caption{Plot of the map $(k,t)\mapsto\mathcal{L}w(k,t)$ (left) 
  in the non-SVI case of Section~\ref{sec:nonSVI2}, 
  and of the density at time $t=1$ (right), with $\nu=3.5$ and $\alpha=1$.
 Here the density does not have a spike at the origin.
 }
  \end{center}
  \end{figure}



\appendix
\section{Proofs}
\subsection{Proof of Proposition~\ref{prop:nomass}}\label{sec:prop:nomass}
The functions~$p_-$ and $p_+$ are clearly well-defined and non-negative. 
Consider first $p_-$. 
It is readily seen that the function  $D(k) \equiv \partial_K \BS(K, v(\log(K)))|_{K=\E^{k}}$
is a primitive of $p_-$. We now proceed to prove that~$p_-$ is indeed a density.
Let $\Nn$ denote the cumulative distribution function of the standard Gaussian distribution.
An explicit computation yields (the reverse one can be found in~\cite[Lemma 2.2]{gj})
$$
\partial_K \BS(K, v(\log(K)))
  =  \frac{\E^{-d_+^2/2}\partial_K d_+}{\sqrt{2\pi}} - \Nn (d_-)
   - \frac{\E^{-d_-^2/2}K\partial_K d_-}{\sqrt{2\pi}}
 = \frac{\E^{-d_+^2/2}}{\sqrt{2\pi}}\left(\partial_K d_+-\partial_K d_-\right)-\Nn (d_-),
$$
where $d_\pm$ and their derivatives are evaluated at $(\log(K),v(\log(K)))$, 
and where we have used the identity $K\Nn'(d_-(\cdot))=\Nn'(d_+(\cdot))$.
Evaluating the right-hand side at $K=\E^{k}$, using 
$-k-\frac{1}{2}d_+^2=-\frac{1}{2}d_-^2$, we obtain
$$
D(k) = \frac{v'(k)}{2\sqrt{2\pi v(k)}}\exp\left(-k - \frac{d_-(k,v(k))^2}{2}\right)-\Nn (d_-).
$$
Therefore if
\begin{equation}\label{eq:limvprime}
\lim_{k\to\pm\infty} \frac{v'(k)}{2\sqrt{2\pi v(k)}}\exp\left(-k - \frac{d_-(k,v(k))^2}{2}\right) = 0, 
\end{equation}
then 
$$
\int_{\RR} p_-(k) \D k = 
\lim_{k\downarrow-\infty}\Nn (d_-(k,v(k)))-\lim_{k\uparrow \infty}\Nn (d_-(k,v(k)))
=1,
$$
where we have used the SMB Condition in Assumption~(II) and Lemma~\ref{lem:RT10}(i).
We now prove~\eqref{eq:limvprime}, and consider first the case when~$k$ tends to (positive) infinity. 
From Lemma~\ref{lem:RT10}(i), $\exp\left(-k-\frac{1}{2}d_-(k,v(k))^2\right)$ tends to zero.
The key point is that $D$ is the primitive of a non-negative function, 
therefore is non-decreasing with a (generalised) limit~$L\in (-\infty,\infty]$ as
$k\uparrow\infty$. 
Since $\Nn \left(d_-(k,v(k))\right)$ converges to zero by Lemma~\ref{lem:RT10}(ii), we deduce that
$\frac{v'(k)}{2\sqrt{2 \pi v(k)}}\exp\left(-k-\frac{1}{2}d_-(k,v(k))^2\right)$ also converges to~$L$.
From~\cite[Proof of Theorem 5.3]{RT10}, 
the inequality $v'(k)<\sqrt{2 v(k)/k}$ holds for any $k>0$ so that $L$ is necessarily non-positive. 
Assume that $L$ is negative; 
since $v'(k)/(2 \sqrt{v(k)}) \equiv \partial_k \sqrt{v(k)}$, 
then $\sqrt{v}$ is eventually decreasing. 
Since it is bounded from below by zero, 
there exists a sequence $(k_n)_{n\geq 0}$ going to infinity such that $\partial_k\sqrt{v(k_n)}$ converges to zero by the mean value theorem, and hence~$L=0$.

Let us now consider the case where~$k$ tends to negative infinity. 
Using similar arguments, the quantity $\frac{v'(k)}{2\sqrt{2 \pi v(k)}}\exp\left(-\frac{1}{2}d_-(k,v(k))^2\right)$ 
tends to~$M\in [-\infty,\infty)$.
Assume that $M<0$. 
Then $v$ is decreasing for~$k$ small enough. 
Since $v$ is positive, this implies that $v(k)>\varepsilon$ for some~$\varepsilon>0$ and~$k$ small enough.
In particular $1/v(k)$ is bounded. 
Since $v'(k)>-4$ for all $k\in\RR$ by~\cite[Theorem 5.1]{RT10}, then for~$k$ small enough, 
the inequalities $-4< v'(k)\leq 0$ hold, and the term outside the exponential in~\eqref{eq:limvprime} is bounded. 
Since the exponential converges to zero by Lemma~\ref{lem:RT10}(ii), we obtain~$M=0$.
If $M>0$, then~$v$ is increasing for $k$ small enough. 
We conclude as above by the mean value theorem since $\sqrt{v}$ is increasing and bounded from below.
Therefore~$M=0$ and the limit~\eqref{eq:limvprime} holds.

So far we have proved that $p_-$ is the density of probability associated to Call option prices with implied volatility $k\mapsto v(k)$.
Consider now the function $w(k)\equiv(-k)$. 
Then for all~$k\in\RR$, $\partial_kw(k)=-\partial_kv(-k)$, and it follows by inspection that
$\Ll w(k)=\Ll v(-k)\geq 0$.
Consider the function~$\widehat{p}_-$ associated to $w$, i.e.
$$
\widehat{p}_-(k) := \left(2 \pi w(k)\right)^{-1/2}\exp\left(- \frac{1}{2}d_{-}^2(k,w(k))\right) \Ll w(k),
\quad\text{for all }k\in\RR.
$$
Now $d_{-}(k,w(k)) \equiv -d_{+}(-k,v(-k))$,
so that $\widehat{p}_-(k) \equiv p_+(-k)$.
In order for $\widehat{p}_-$ to be a genuine density, 
we need to check conditions symmetric to those ensuring that~$p_-$ is a density.
The condition symmetric to the SMB assumption~(II) is precisely Condition~(i) in Lemma~\ref{lem:RT10},
and the condition symmetric to the Lemma~\ref{lem:RT10}(ii) is precisely the LMB assumption~(III).
Therefore $k\mapsto p_+(-k)$ is also a density, associated to a Call option price with implied variance~$w$.
Finally the identity~$p_+(k) = \E^{k} p_-(k)$ follows immediately from the equality
$-k-\frac{1}{2}d_+^2=-\frac{1}{2}d_-^2$.

\subsection{Proof of Proposition~\ref{prop:noCSAStrictly}}\label{sec:prop:noCSAStrictly}
Assume that $\Psi$ is asymptotically linear and that there is no calendar spread arbitrage.
The proof relies on the decomposition of the real line into the disjoint unions
$\RR = \overline{\Zz}_+ \cup \left(\overline{\Zz}_-\cap \Zz_+(u)\right) \cup \left(\overline{\Zz}_-\cap (\RR\setminus \Zz_+(u))\right)$, for any $u>0$.
As in the proof of Proposition~\ref{prop:Coupling2}, butterfly arbitrage is precluded on $\RR\setminus \Zz_+(u)$, 
so that we are left with~$\overline{\Zz}_+$ and~$\overline{\Zz}_-\cap \Zz_+(u)$.

Consider first Case~(ii).
If $z \in \aleph^c$, the inequality in the proposition follows from~\eqref{ineq:condfly}.
When $z\in\aleph$, in view of~\eqref{ineq:condfly}, the inequality $\Ll w(k,t)\geq 0$ is equivalent to
$(\theta_t \varphi(\theta_t))^2 \Psi''(z) \geq -2$.
Since $\Psi''$ is strictly negative on $\overline{\Zz}_+\cap\aleph$, and non-negative on $\overline{\Zz}_-\cap\aleph$, absence of butterfly arbitrage on $\aleph$ is equivalent to
\begin{equation*}
\displaystyle (\theta_t \varphi(\theta_t))^2 \leq -\frac{2}{\Psi''(z)}
\text{ on }\overline{\Zz}_+\cap\aleph
\qquad\text{and}\qquad
\displaystyle (\theta_t \varphi(\theta_t))^2 \geq -\frac{2}{\Psi''(z)}
\text{ on }\overline{\Zz}_-\cap\aleph,
\end{equation*}
where the inequalities are trivial (bounds equal to $\pm\infty$) 
whenever $\Psi''(z)=0$.
In fact, on $\overline{\Zz}_-\cap\aleph$, this inequality is trivially satisfied, and the result holds.

Consider now Case~(i) in the proposition, which corresponds to the set $\overline{\Zz}_-\cap \Zz_+(u)$.
We can in fact restrict our attention to $\overline{\Zz}_-\cap \Zz_+(u)\cap\aleph^c$ since $\overline{\Zz}_-\cap \Zz_+(u)\cap\aleph$ is empty;
the map $u \mapsto u \varphi(u)$ is non-decreasing on~$\RR_+^*$ by Proposition~\ref{prop:first_lin}.
On $\overline{\Zz}_-\cap\aleph^c$ 
the map $u \mapsto \frac{1}{4 u} (\frac{\Psi'(z)^{2}}{\Psi(z)}-2\Psi''(z))+\frac{\Psi'(z)^{2}}{16}$
is clearly also non-decreasing on~$\RR_+^*$.
Therefore $(\overline{\Zz}_{-} \cap \Zz_+(u)\cap\aleph^c)_{u>0}$ is a non-decreasing family of sets 
and thus, in view of~\eqref{ineq:condfly}, 
absence of butterfly arbitrage ($\mathcal{L}w\geq 0$) on this set is equivalent to 
$$
(u \varphi(u))^2 \leq \inf_{z \in \overline{\Zz}_- \cap \Zz_+(u)\cap\aleph^c} 
\Lambda(u,z),
\quad\text{for all }u \in (0, \theta_\infty),$$
when $\theta_\infty<\infty$, which in turn is equivalent to
$$ 
M_\infty^2 \leq \inf_{z \in \overline{\Zz}_{-} \cap \Zz_+(\theta_\infty)\cap\aleph^c}
 \Lambda(\theta_\infty,z).
$$
When $\theta_\infty=\infty$, the previous infimum is precisely 
$\inf\{|\frac{4}{\Psi'(z)}-\frac{2z}{\Psi(z)}|, 
z\in \overline{\Zz}_- \cap \left(\cup_{u>0} \Zz_+(u)\right)\cap\aleph^c \} $. 
Now, the set 
$\overline{\Zz}_- \cap \left(\cup_{u>0} \Zz_+(u)\right)\cap\aleph^c = \overline{\Zz}_{-}\cap\aleph^c$
is not empty, and therefore the last upper bound is also equal to 
$\inf\{|\frac{4}{\Psi'(z)}-\frac{2z}{\Psi(z)}|, z\in \overline{\Zz}_- \cap\aleph^c\} $.

We note in passing that~$\overline{\Zz}_+$ is not empty.
Otherwise, the asymptotic linearity of $\Psi$ allows us to choose $a > 0$ such that $\Psi'(z) > 0$ for all $z>a$.
Therefore $1/\Psi(z)\leq 2\Psi''(z)/\Psi'(z)^2$  for all $z>a$,
which in turn yields $\int_a^z \frac{\D b}{\Psi(b)} \leq 2(\Psi'(a)^{-1}-\Psi'(z)^{-1})$.
The integral diverges to infinity as~$z$ tends to infinity since $\Psi(z) \sim \alpha_+ z$
whereas the right-hand side is bounded by Definition~\ref{def:Linear}. 
The same argument shows that~$\overline{\Zz}_+$ is not bounded from above.

\subsection{Proof of Proposition~\ref{prop:EffBoundPsi}}\label{sec:prop:EffBoundPsi}
In the generalised SVI case~\eqref{eq:GenSVI}, the function~$\Psi$ is asymptotically linear 
(see Definition~\ref{def:Linear}) with 
$\lim_{z\uparrow \infty}\Psi'(z)=\alpha_+>0$, and $\theta_\infty = \infty$.
From subsection~\ref{sec:Linear} the condition $M_{\infty} \leq \left| \frac{4}{\Psi'(z)} - \frac{2 z}{\Psi(z)} \right|$
holds for all $z \in \Zz_+(\theta_\infty) = \mathbb{R}$.
Since $\lim_{z\uparrow\infty}\left(\frac{4}{\Psi'(z)} - \frac{2 z}{\Psi(z)}\right)=\frac{2}{\alpha_+}$, we can define
$$ z_+:= \inf \left\{z \in \mathbb{R}_+: \inf_{y\geq z}\left(\frac{4}{\Psi'(y)} - \frac{2 y}{\Psi(y)}\right) > 0 \right\}<\infty,$$
and therefore
\begin{equation}\label{eq:Inequal}
\frac{4}{\Psi'(z)} - \frac{2 z}{\Psi(z)} \geq M_{\infty},
\quad\text{ for all }z \geq z_+.
\end{equation}
Note that $M_\infty \leq \frac{2}{\alpha_+}$, and let $u_+ := \Psi(z_+)$.
Since the continuous function~$\Psi$ is increasing on $[z_+,+\infty)$, 
we can define its inverse~$g:[u_+,+\infty) \to [z_+,+\infty)$,
and hence from the equality $\exp\left(-\int_{u_+}^{u} \frac{\D v}{2 v}\right) = \sqrt{\frac{u_+}{u}}$, 
Equation~\eqref{eq:Inequal} reads
\begin{align*}
\frac{4}{\Psi'(z)} - \frac{2 z}{\Psi(z)} \geq M_{\infty}
 & \iff g'(u)-\frac{g(u)}{2 u} \geq \frac{M_{\infty}}{4}\\
 & \iff \partial_u\left(g(u) \exp\left(-\int_{u_+}^{u} \frac{dv}{2 v}\right)\right) \geq \frac{M_{\infty}}{4} \sqrt{\frac{u_+}{u}}\\
 & \iff g(u) \sqrt{\frac{u_+}{u}} - g(u_+) \geq \frac{M_{\infty}}{2} \sqrt{u_+} \left(\sqrt{u} - \sqrt{u_+}\right)\\
 & \iff g(u) \geq g(u_+) \sqrt{\frac{u}{u+}} + \frac{M_{\infty}}{2} \sqrt{u} \left(\sqrt{u} - \sqrt{u_+}\right).
\end{align*}
where all the inequalities on the right-hand side are considered for $u\geq u_+$.
The third line is obtained by integration between $u_+$ and $u$ on both sides of second line.
Let $K_l := \frac{1}{2}M_{\infty}$ and $K_s := u_+^{-1/2}g(u_+) - \frac{1}{2}M_{\infty}{\sqrt{u_+}} $.
We then obtain the condition
\begin{equation}\label{eq:Inequal3}
g(u) \geq K_s \sqrt{u} + K_l u, 
\qquad\text{for all }u\geq u_+.
\end{equation}
Note that $K_s$ remains non-negative if we increase $z_+$ or decrease $M_\infty$;
indeed $\lim_{z\uparrow\infty}(2 z/\Psi(z)) =2/\alpha_+$, so that the condition
$M_\infty \leq 2/\alpha_+$ is equivalent to
$M_\infty \leq 2 z/\Psi(z) = 2 g(u)/u$.
Finally let us translate condition~\eqref{eq:Inequal3} into conditions on~$\Psi$.
Fix $u\geq u_+$ and denote $z:=g(u)$, then 
$(z - K_l u)^2 \geq K_s^2 u$, 
which is equivalent to $K_l^2 u^2 - (K_s^2+2 K_l z) u + z^2 \geq 0$.
The discriminant is equal to $K_s^2(K_s^2+4 z K_l)$ and is clearly non-negative.
Condition~\eqref{eq:Inequal3} is therefore equivalent to
$$
\Psi(z) \notin \left[\frac{K_s^2+2 K_l z - K_s \sqrt{K_s^2+2 K_l z}}{2 K_l^2}, \frac{K_s^2+2 K_l z + K_s \sqrt{K_s^2+2 K_l z}}{2 K_l^2}\right].
$$
Given $z-K_l u\geq 0$ (equivalently $\Psi(z) \leq z/K_l$) we obtain
$\Psi(z) \leq \kappa^2+\lambda z - \kappa \sqrt{\kappa^2+\lambda z},$
where 
$\kappa:=K_s/\left(\sqrt{2}K_l\right)$ and $\lambda:=K_l^{-1}$.

\subsection{Proof of Proposition~\ref{prop:nonSVI2}}\label{sec:prop:nonSVI2}
The function~$f$ defined in~\eqref{eq:Ff} therefore reads $f(u) = \frac{(u+1) \E^{-u} - 1}{1-\E^{-u}}$, 
with $f(0)=0$
and is strictly decreasing from~$0$ to~$-1$.
Regarding the function~$F$, it is clearly continuous, increasing from $0$ to $1$ and 
$F(z)=|z|^\nu/\left(1+|z|^\nu\right)$ for all $z\ne 0$,
with $F(0)=0$.
Since $\theta_\cdot$ is increasing and $1+f(u)F(z)\geq 0$ for all $(u,z)\in\RR_+^*\times\RR$, the first coupling conditions in Proposition~\ref{prop:first_coupling} are satisfied, and the volatility surface is free of calendar spread arbitrage.
We now need to check the second coupling condition, namely Proposition~\ref{prop:Coupling2}. For $\nu\geq2$, $\Psi$ is $\mathcal{C}^2$, and we can indeed apply Proposition~\ref{prop:Coupling2}. 
Since~$\Psi$ is asymptotically linear, we can alternatively check Proposition~\ref{prop:noCSAStrictly}.
The equality 
\begin{equation}\label{eq:Phi}
\Phi_\nu(z)\equiv \frac{\Psi'_\nu(z)^2}{\Psi_\nu(z)}-2\Psi_\nu''(z) = 
\left(1+|z|^\nu\right)^{1/\nu-2}|z|^{\nu-2}\left(|z|^\nu-2(\nu-1)\right)
\end{equation}
holds for all $z\ne 0$ and hence the sets~$\overline{\Zz}_+$ and~$\overline{\Zz}_-$ 
defined in~\eqref{eq:SetZBarP} are equal to 
$\overline{\Zz}_-=[z_-^*,z_+^*]$
and $\overline{\Zz}_+=\RR\setminus [z_-^*,z_+^*]$,
where $z_\pm^*:=\pm [2(\nu-1)]^{1/\nu}$.
The two conditions in Proposition~\ref{prop:noCSAStrictly} read
\begin{align}
M_\infty^2 
& \displaystyle \leq \inf_{z \in \overline{\Zz}_- \bigcap \Zz_+(\theta_\infty)}
\frac{\left(1-\frac{z\Psi_\nu'(z)}{2\Psi_\nu(z)}\right)^2}
{\frac{1}{4 \theta_\infty} \Phi_n(z)+\frac{\Psi_\nu'(z)^{2}}{16}};\nonumber\\
(u \varphi(u))^2 
& \displaystyle \leq \inf_{z \in \overline{\Zz}_+} 
\frac{\left(1-\frac{z\Psi_\nu'(z)}{2\Psi_\nu(z)}\right)^2}
{\frac{1}{4 u} \Phi_n(z)+\frac{\Psi_\nu'(z)^{2}}{16}},
\quad\text{for any }u \in \RR_+^*.\label{eq:Ineq2}
\end{align}
From the proof of Proposition~\ref{prop:noCSAStrictly}, we know that when $\theta_\infty=\infty$, 
the first condition simplifies to
\begin{equation}\label{eq:Ineq1}
M_\infty \leq \inf_{z\in [z_-^*,z_+^*]}\left|\frac{4}{\Psi_\nu'(z)}-\frac{2z}{\Psi_\nu(z)}\right|.
\end{equation}
Now, immediate computations yield
$\displaystyle
\left|\frac{4}{\Psi_\nu'(z)}-\frac{2z}{\Psi_\nu(z)}\right|
= \left|\frac{2z}{|z|^\nu}\frac{2+|z|^\nu}{\left(1+|z|^\nu\right)^{1/\nu}}\right|,
$
which, as a function of $z$ is defined on $\RR^*$, is strictly increasing on~$\RR_-^*$ 
and strictly decreasing on~$\RR_+^*$.
Therefore, its infimum~$z_\nu$ over the interval $[z_-^*,z_+^*]$ is precisely attained at~$z_\pm^*$ 
(by symmetry) and is equal to
$4\nu(2\nu-2)^{(1-\nu)/\nu}(2\nu-1)^{-1/\nu}$.
Since by construction $M_\infty=\alpha$, Inequality~\eqref{eq:Ineq1} is thus equivalent to $\alpha \leq z_\nu$.
This inequality is clearly not true for any $\nu>1$ and $\alpha>0$;
however straightforward considerations show that there exists a unique $\nu^*>1$ 
such that the map $\nu\mapsto z_\nu$ is strictly increasing on $(1,\nu^*)$ 
and strictly decreasing on $(\nu^*,\infty)$ with $z_1=4$ and $\lim_{z\uparrow \infty} z_\nu = 2$.
Therefore the inequality $\alpha \leq z_\nu$ is satisfied for all $\nu>1$ if and only if $\alpha\leq 2$.

We now check the second inequality~\eqref{eq:Ineq2} above.
Straightforward computations show that
$$
\left(1-\frac{z\Psi_\nu'(z)}{2\Psi_\nu(z)}\right)^2
 = \frac{1}{4}\left(\frac{2+|z|^\nu}{1+|z|^\nu}\right)^2,
$$
which increases on~$\RR_-$ from~$\frac{1}{4}$ to~$1$ 
and decreases on~$\RR_+$ from~$1$ to~$\frac{1}{4}$.
The map $\Psi_\nu'(\cdot)^{2}$ is decreasing on~$\RR_-$, increasing on~$\RR_+$ 
and maps the real line to~$(0,1)$.
Therefore, for any $u>0$, $z \in \overline{\Zz}_+$, we have
$$
\left(\frac{1}{4 u}\Phi_n(z)+\frac{\Psi_\nu'(z)^{2}}{16}\right)^{-1}
\left(1-\frac{z\Psi_\nu'(z)}{2\Psi_\nu(z)}\right)^2
\geq 
\frac{1}{4}\frac{1}{\frac{1}{4 u}\Phi_\nu(z)+\frac{1}{16}}
= \left(\frac{\Phi_\nu(z)}{u}+\frac{1}{4}\right)^{-1},
$$
with $\Phi_\nu$ defined in~\eqref{eq:Phi}.
Now a quick look at the function $\Phi_\nu$ shows that 
it is bounded above by $\Phi_\nu(z_\nu^*) \in (0,1)$, 
with $z_\nu^*:=[\nu(\nu-1)-2+\sqrt{\nu(\nu-1)(\nu^2+3\nu-2)}]^{1/\nu}$.
Define the function $g_\alpha$ by
$g_\alpha(u)\equiv (u\varphi(u))^2 \left(\frac{1}{u}+\frac{1}{4}\right)$.
There exists a unique $u^*\approx 1.87$ such that~$g_\alpha$ is strictly increasing
on $(0,u^*)$ and strictly decreasing on $(u^*,\infty)$ with $g_\alpha(u^*) = g_1(u^*)\alpha^2$.
Setting $\overline{\alpha}:=g_1(u^*)^{-1/2}\approx 1.33$, the inequality $g_\alpha(u)\leq 1$
is clearly satisfied for any $u>0$ and all $\alpha \in (0,\overline{\alpha})$.
To conclude, note that for $1<\nu<2$, the second derivative has a mass at the origin, 
but $\Psi_\nu$ is convex which implies that this mass is positive and 
that $\Ll w(k,t)\geq0$ in the distributional sense following Section~\ref{sec:nonsmooth}.
Therefore the implied volatility surface is free of static arbitrage and the proposition follows.

\subsection{Proof of Proposition~\ref{prop:nonSVI1}}\label{sec:prop:nonSVI1}
The function~$f$ in~\eqref{eq:Ff} here reads $f(u) = \frac{(u+1) \E^{-u} - 1}{1-\E^{-u}}$, 
and is decreasing from~$0$ to~$-1$.
The function~$F$ in~\eqref{eq:Ff} is clearly continuous, increasing from~$0$ to~$1$ and 
$$
F(z)=\frac{|z| (4 \sqrt{1+|z|} + 1)}{2\sqrt{1+|z|}(2 |z|+1+\sqrt{1+|z|})}, 
\qquad\text{for all }z\in\RR^*,
$$
with $F(0)=0$.
By Proposition~\ref{prop:first_coupling}, straightforward computations then show that the volatility surface is free of calendar-spread arbitrage.
Now, for any $z\geq 0$, we have
$$\frac{\Psi'(z)^2}{\Psi(z)}-2\Psi''(z) = \frac{(16 z+19)\sqrt{1+z} + 12 z + 10}{16 (1+z)^{3/2} \Psi(z)},$$
which is a decreasing function of $z$ with limit equal to zero.
Therefore $\overline{\Zz}_+=\RR $, and for any $u\in \RR_+^*$, $\Zz_+(u)=\RR$.
Let us check that the generalised SVI surface~$w$ parameterised by the previous triplet~$(\theta,\varphi,\Psi)$ satisfies $\Ll w \geq 0$ as a distribution. 
Indeed we only checked that $\{\Ll w\}\geq 0$ as a function defined everywhere except at the origin
(where as usual in distribution notations, $\{\Ll w\}$ is a function defined where $w''$ is defined).
Here  $\Psi''= \{\Psi''\} + \frac{5}{2} \delta_0$ (where $\delta_0$ stands for the Dirac mass at the origin), so that
$\Ll w= \{\Ll w\} + 5 (\theta \phi(\theta))^2 \delta_0$, 
which is positive since $\{\Ll w\} \geq 0$.
Finally, 
$$
\frac{4}{\Psi'(z)}-\frac{2 z}{\Psi(z)} = 4 \frac{4 (z+1)^{3/2} + 3 z + 4}
{\left(4\sqrt{z+1}+1\right)\left(2z+1+\sqrt{z+1}\right)}
$$ 
decreases to $2\geq M_\infty$.
Since~$\Psi'(z)^2\geq 1$, the condition~$(u \varphi(u))^2 \leq 4$ suffices to prevent butterfly arbitrage.

\section{Lee's moment formula in the asymptotically linear case}
In Section~\ref{sec:Arbitrage} we stressed that, following Roper or the variant in Proposition~\ref{prop:nomass},
the positivity of the operator~$\Ll$ in~\eqref{eq:OpL} guarantees the existence of a martingale explaining market prices. 
As a consequence, the celebrated moment formula~\cite{Lee} holds:
\begin{theorem}[Roger Lee's moment formula~\cite{Lee}]\label{thm:LeeMoment}
Let $S_t$ represent the stock price at time $t$, assumed to be a non-negative random variable with positive and finite expectation.
Let $\widetilde{p}:=\sup\{p\geq 0: \mathbb{E}(S_t^{1+p})<\infty\}$ and $\beta:=\limsup_{k\uparrow\infty}k^{-1}w(k,t)$.
Then $\beta\in [0,2]$ and $\widetilde{p}=\frac{1}{2}\left(\frac{\beta}{4}-1+\frac{1}{\beta}\right)$.
\end{theorem}
We show here that, at least in the asymptotically linear case (Definition~\ref{def:Linear}), 
this moment formula can be derived in a purely analytic fashion.
\begin{proposition}
Consider a $\mathcal{C}^2(\RR)$ function $v$ satisfying the following conditions:
\begin{enumerate}
\item $v(k)>0$ and $\Ll v(k) \geq 0$ for all $k \in \mathbb{R}$;
\item $\lim_{k\uparrow \infty} v'(k) = \alpha \in (0,2)$;
\item $\lim_{k\uparrow \infty} v''(k)=0$.
\end{enumerate}
Let~$X$ be a random variable with density~$p_-$, associated to~$v$ by Proposition~\ref{prop:nomass}.
Then $\mathbb{E}(X)=1$ and 
$\displaystyle \sup \{m\geq 0 : \mathbb{E}(X^{1+m}) < \infty \} 
= \frac{1}{2}\left(\frac{\alpha}{4}-1+\frac{1}{\alpha}\right).$
\end{proposition}

\begin{proof}
Condition~(1) implies Proposition~\ref{prop:nomass}(I), 
and Conditions~(2) and~(3) imply the SMB and LMB limits in Proposition~\ref{prop:nomass}(II)-(III). 
Therefore by Proposition~\ref{prop:nomass}, the centred probability density~$p_-$ is well defined on $\RR$ and, 
for any $m\in\RR$, we have 
$\E^{(1+m)k} p_-(k) = f(k) \E^{-g(k)}$, where
$$
f(k) \equiv \left(2 \pi v(k)\right)^{-1/2}\Ll v(k)
\qquad\text{and}\qquad
g(k) \equiv \frac{1}{2} \left(\frac{k^2}{v(k)} + \frac{v(k)}{4} + k\right) - (1+m)k.
$$
As $k$ tends to infinity, straightforward computations show that 
$f(k) \sim \frac{4-\alpha^2}{16\sqrt{2\pi \alpha k}}$
and 
$\lim_{k\uparrow \infty}\frac{g(k)}{k}
 = \frac{(\alpha-2)^2-8m\alpha}{8\alpha}:=\frac{P_m(\alpha)}{\alpha}$.
Since~$P_m$ is a second-order strictly convex polynomial with $P_m(0)>0$,
the function $k\mapsto \E^{(1+m)k} p_-(k)$ is integrable as long as $P_m(\alpha)>0$, i.e.
$\alpha<2-4\left(\sqrt{m^2+m}-m\right)$, or
$m<\frac{\alpha}{8}-\frac{1}{2}+\frac{1}{2 \alpha}$.
In other words, we have proved that
$\sup\left\{m>0 : \mathbb{E}\left(X^{1+m}\right)<\infty \right\} 
= \frac{1}{2}\left(\frac{\alpha}{4}-1+\frac{1}{\alpha}\right)$.
\end{proof}


\end{document}